\documentclass[journal,a4paper]{IEEEtran}
\usepackage{graphicx}
\usepackage{amsmath,amssymb}
\usepackage{amssymb}
\usepackage{amsthm,amsmath}
\usepackage{lineno}
\usepackage{verbatim}
\newtheorem{thm}{Theorem}
\newtheorem{lem}{Lemma}
\newtheorem{mydef}{Definition}
\newtheorem{asu}{Assumption}
\newtheorem{col}{Corollary}
\newtheorem{rem}{Remark}
\allowdisplaybreaks
\begin{document}
%
\title{Adaptive Finite Time Stability of Delayed Systems With Applications to Network Synchronization}

\author{Xiwei Liu,~\IEEEmembership{Senior member, IEEE}
\thanks{This work was supported by the National Science Foundation of China under Grant No. 61673298, 61203149; Shanghai Rising-Star Program of China under Grant No. 17QA1404500; Natural Science Foundation of Shanghai under Grant No. 17ZR1445700; the Fundamental Research Funds for the Central Universities of Tongji University.}
\thanks{Department of Computer Science and Technology, Tongji University, and with the Key Laboratory of Embedded System and Service Computing, Ministry of Education, Shanghai 201804, China. Corresponding author: Xiwei Liu. Email:xwliu@tongji.edu.cn}}

\maketitle
\begin{abstract}
The finite time stability (FnTSta) theory of delayed systems has not been set up until now. In this paper, we propose a two-phases-method (2PM), to achieve this object. In the first phase, we prove that the time for norm of system error evolving from initial values to 1 is finite; then in the second phase, we prove that the time for norm of the error evolving from 1 to 0 is also finite, thus FnTSta is obtained. Considering the cost and complexity of controller, we use only two simple terms to realize this aim. For the proposed 2PM, time delays can be the same or asynchronous, bounded or unbounded, etc; the norm can be $2$-norm, $1$-norm, $\infty$-norm, which show that 2PM is powerful and has a wide scope of applications. Furthermore, we also prove the adaptive finite time stability (AFnTSta) theory of delayed systems using 2PM. As an application of the obtained FnTSta theory, we consider the finite time outer synchronization (FnTOSyn) for complex networks with asynchronous unbounded time delays, corresponding criteria are also obtained. Finally, two numerical examples are given to demonstrate the validity of our theories.
\end{abstract}

\section{Introduction}
Convergence rate is an important index for investigating the dynamics of system equilibrium. The most studied cases are exponential stability and asymptotic stability by using the negative feedback technique, where the state of controlled system will approach to the equilibrium as close as possible but never be equal to the equilibrium. Compared with this infinite process, finite time stability (FnTSta) is proposed \cite{BB1998}-\cite{MP2008}, which means that the system will reach the equilibrium in FnT (called the settling time). Of course, enhancement of the final target requires the redesign of external control, i.e., linear feedback can realize exponential stability, while FnTSta must use nonlinear feedback, and the relationship between time and control energy was discussed in \cite{linwei2017}. Especially, in \cite{BB2000}, the authors used
\begin{align}
\dot{V}(t)\le -cV(t)^{\alpha},~ 0<\alpha<1
\end{align}
to investigate the FnTSta, which now has been the standard method of most papers in this area. In FnTSta analysis, the settling time depends on initial values. If the settling time is independent of initial values, it is called the fixed time stability (FxTSta) \cite{P2012}.

Synchronization/consensus for complex networks, including neural networks (NNs) and multi-agent systems (MAS), as a typical example of cooperation between coupled systems has attracted more and more attention, whose development process also evolves from exponential/asymptotic synchronization \cite{Murray2004, Chenliulu2007} to FnTSyn/FxTSyn along with the development of stability theory. \cite{c2006} considered FnTSta using the normalized and signed gradient and applied it to network consensus; \cite{jiangwang2011} investigated FnT consensus using a monotonic function; \cite{yangtac2016} investigated FnTSyn with markovian topology and impulsive effects; \cite{licaojiang2019} studied FnT fractional-order networks synchronization; \cite{p2015}-\cite{liuchen2018} set up the general theory for FnTSyn/FxTSyn; \cite{huyunn2017} investigated FxTSyn with discontinuous neural networks; \cite{liuxiaoyang2019} studied FnT/FxT pinning synchronization for networks with stochastic disturbances; \cite{zuoii2018} presented a review of FxTSyn; \cite{rios2017} designed FnT/FxT parameter identification algorithms.

In the real world, time delay is an important factor for consideration. Firstly, it arises in many cases, such as the ability of information processing unit is limited, the long distance in information transmission exists, or the outbreak of the virus has a incubation period, etc. Secondly, there are many types of delays, such as constant delay \cite{liuchen2016}, time-varying delay, bounded delay, unbounded delay \cite{chenmu}, distributed delay, etc. Thirdly, the existence of delay may impact even change the dynamical behaviors of systems. Therefore, there have been various works investigating the FnTSta or FnTSyn for delayed systems. \cite{wancaonn2016} investigated the robust FxTSyn for master-slave coupled delayed Cohen-Grossberg NNs by designing an external control with time delays (equation (13) in \cite{wancaonn2016}); \cite{liujianghu2017} studied the FnTSyn with aperiodically intermittent control (AIC) by designing a control with time delays (equation (10) in \cite{liujianghu2017}); \cite{guogonghuang2018} and \cite{weicaocn2018} investigated FnTSyn of inertial memristive delayed NNs by designing a control with time delays (equation (6) in \cite{guogonghuang2018} and equation (12) in \cite{weicaocn2018}); \cite{ganxaosheng2019} studied the FxT outer synchronization for networks with aperiodically semi-intermittent control by designing a control with time delays (equation (2.5) in \cite{ganxaosheng2019}).

Therefore, based on above discussions, we would raise {\bf Problem 1: ``For delayed systems, can we design a simple controller without using the time delay to realize FnTSta? How to prove its validity?''}

On the other hand, adaptive technique is important in real applications. In many cases, the parameters cannot be obtained directly and parameters may also be varying, therefore, a constant control coefficient may be not enough, while adaptive technique can solve this trouble and has been widely used. For example, in \cite{yanli2019}, the authors realized the adaptive exponential synchronization for AIC. Therefore, if Problem 1 is solved, we would like to propose a more difficult {\bf Problem 2: ``For delayed systems, can we design a simple adaptive rule, such that FnTSta can be realized? How to prove its validity?''}

Now, we introduce some excellent related works about the two problems. \cite{jingzhangmeifan2019} studied FnTSyn via AIC, where the time delay was constant and the external control contained only two terms, by designing a Lyapunov function using this constant delay, the authors proved a useful FnT lemma for AIC. \cite{liuqin2019} considered the drive-response FnTSyn for complex-valued NNs with bounded and differentiable time-varying delays, and distributed delays,  designed a new Lyapunov-Krasovskii function using $1$-norm, and proved its derivative was less than a negative constant, thus the FnTSyn was proved. \cite{yangprop} investigated the FnTSyn for memristive NNs with proportional delay $qt$, by using Filippov solution and Lyapunov-Krasovskii functional method, FnTSyn was proved. It is worth mentioning that an adaptive rule was also designed in \cite{yangprop}. The limitation of these previous works is that time delays are required to be concrete in order to construct Lyapunov-Krasovskii functional, like constant \cite{jingzhangmeifan2019}, or bounded and differentiable \cite{liuqin2019}, or proportional \cite{yangprop}; moreover, time delays are required to be the same for each node.

Motivated by the above analysis, in this paper, we will address the above two problems. The contributions and adavantages of this paper can be summarised as follows:

{\bf (1): The time delay for FnTSta has a large allowable range.} It can be bounded or unbounded, the same or asynchronous, differentiable or not differentiable, etc. To achieve this aim, we will use the maximum-value method \cite{chenmu,liuchen2016}.

{\bf (2): Two-phases-method (2PM) is shown to be powerful in FnT problems.} Two phases means that the proof process is divided into two phases. In the first phase, the norm of system error will be proved to evolve from initial values to $1$ in FnT (of course, $1$ can be replaced by any constant). In the second phase, the norm of system error will be proved to evolve from $1$ to $0$ in FnT. Related works about 2PM can be found in \cite{wangchen2018,wangchen2019,liuli2020}, where the external control are composed of three terms to realize FnT anti-synchronization, while in this paper the number of control terms is reduced to two.

{\bf (3): The FnTSta theory for delayed systems is set up using 2PM.} Therefore, problem 1 has been completely solved. The corresponding theorems under $2$-norm, $1$-norm and $\infty$-norm are presented respectively. Moreover, the FnTSta theory is applied to FnT network outer synchronization problem.

{\bf (4): The AFnTSta theory for delayed system is set up using 2PM.} Therefore, problem 2 has been completely solved. The corresponding theorems and adaptive rules under three norm are given respectively. Moreover, AFnTSta theory is also applied to FnT network outer synchronization problem.

The rest of this paper is organized as follows. In Section \ref{FnTSta}, the FnTSta theory and AFnTSta theory for a delayed simple system model using 2PM are set up under three different norms respectively. In Section \ref{FnTSyn}, we apply the obtained theory under $2$-norm to the FnT network outer synchronization. Moreover, two adaptive schemes are also designed to realize FnTSyn. In Section \ref{ns}, we present two numerical examples to illustrate the validity of our obtained theoretical results. Finally, we conclude this paper in Section \ref{con}.

\section{Theories of finite time stability with delay}\label{FnTSta}
\subsection{A simple model}
At first, we set up some theories about (adaptive) FnTSta with unbounded (including bounded) time-varying delays of a general model:
\begin{align}\label{const}
\dot{p}(t)=&c_1p(t)+c_2p(t,\Pi(t))\nonumber\\
&-\mathrm{diag}(\mathrm{sgn}(p(t)))(c_3{\bf 1}+c_4|p(t)|)
\end{align}
where $p(t)=(p_1(t),\cdots,p_m(t))^T\in R^{m\times 1}$,
\begin{align*}
p(t,\Pi(t))=&(p_1(t-\pi_1(t)),\cdots,p_m(t-\pi_m(t)))^T,\\
\mathrm{sgn}(p(t))=&(\mathrm{sign}(p_1(t)),\cdots,\mathrm{sign}(p_m(t)))^T, \\
|p(t)|=&(|p_1(t)|,\cdots,|p_m(t)|)^T, \\
{\bf 1}=&(1,\cdots,1)^T\in R^m.
\end{align*}

\begin{asu}\label{as1}
As for the asynchronous unbounded time delay $\pi_i(t)$, suppose we can find a continuous function $\pi(t)$ such that
\begin{align}
\pi_i(t)\le\pi(t)
\end{align}
where $t-\pi(t)\rightarrow+\infty$ and $\pi(t)\rightarrow+\infty$ as $t\rightarrow +\infty$.
\end{asu}

\begin{asu}\label{as2}
For this $\pi(t)$, suppose we can find a nondecreasing function $\mu(t)$ such that $\lim_{t\rightarrow+\infty}\mu(t)=+\infty$, and
\begin{align}\label{mu}
\overline{\lim\limits_{t\rightarrow +\infty}}\frac{\dot{\mu}(t)}{\mu(t)}=\beta, \overline{\lim\limits_{t\rightarrow+\infty}}\frac{\mu(t)}{\mu(t-\pi(t))}=1+\eta\ge 1
\end{align}
\end{asu}

\begin{rem}
In fact, system (\ref{const}) can be simply transformed as
\begin{align*}
\dot{p}(t)=&(c_1-c_4)p(t)+c_2p(t-\pi(t))-c_3\mathrm{diag}(\mathrm{sgn}(p(t)))\cdot{\bf 1}
\end{align*}
which makes the equation more simple. However, we still take the form as (\ref{const}). The reason is that the two terms in the first line of (\ref{const}) can be regarded as the intrinsic/original function of the system, while the two terms in the second line of (\ref{const}) can be regarded as the external control on the system.
\end{rem}

\begin{mydef} (see \cite{BB2000})
The origin $p(t)=0$ is said to be a finite time stable equilibrium if the finite time convergence condition and Lyapunov stability condition hold globally.
\end{mydef}

The main method in this paper is two-phases-method (2PM), which is explicitly called in \cite{liuli2020}, its validity will be shown by using $2$-norm, $1$-norm, and $\infty$-norm respectively.

\subsection{FnTSta under $2$-norm by using 2PM}
In this subsection, we firstly choose the $2$-norm, due to the wide applications in the literature of synchronization for complex networks. Moreover, under $2$-norm, one can use the property of coupling matrix to show that pinning control can realize the synchronization.
\begin{thm}\label{const2normthm}

If there exists a positive constant $\varepsilon_1$, such that the following conditions hold:
\begin{align}
&\beta+2(c_1-c_4)+|c_2|\varepsilon_1+|c_2|m\varepsilon_1^{-1}(1+\eta)<0\label{c1}\\
&|c_2|-c_3<0\label{c2}
\end{align}
where $\beta$ and $\eta$ are defined in (\ref{mu}), then FnTSta with unbounded time delays for (\ref{const}) will be realized, where the settling time will be defined later.
\end{thm}
\begin{proof}
From (\ref{c1}) and Assumption \ref{as2}, we can get a time $T$, such that for $t\ge T$,
\begin{align}\label{transf}
\frac{\dot{\mu}(t)}{\mu(t)}+2(c_1-c_4)+|c_2|\varepsilon_1+|c_2|m\varepsilon_1^{-1}\frac{\mu(t)}{\mu(t-\pi(t))}<0
\end{align}

In the following, we will always discuss from $T$. 2PM will be shown carefully in the next proof procedure. If $\sup_{T-\pi(T)\le s\le T}\|p(s)\|_{2}\le 1$, then we will switch to Phase II directly; else, we will start from Phase I.

{\bf Phase I:} At first, we will prove that the norm of $p(t)$ will decrease from initial values to $1$ in FnT.

Define a Lyapunov function
\begin{align}\label{v1}
V_1(t)=\mu(t)p(t)^Tp(t)
\end{align}
and a maximum-value function
\begin{align}\label{w1}
W_1(t)=\sup\limits_{t-\pi(t)\le s \le t}V_1(s)
\end{align}
Obviously, $V_1(t)\leq W_1(t)$, $t\ge T$.
	
If $V_1(t)<W_1(t)$. Then there must exist a constant $\delta_1>0$ such that $V_1(s)\le W_1(t)$ for $s\in (t,t+\delta_1)$, i.e., $W_1(s)=W_1(t)$ for $s\in (t,t+\delta_1)$.
	
Else if there exists a point $t_1\ge T$, $V_1(t_1)= W_1(t_1)$, then
\begin{align}
&\frac{dV_1(t)}{dt}\bigg{|}_{t=t_1}=\dot{\mu}(t)p(t)^Tp(t)+2\mu(t)p(t)^T\dot{p}(t)\nonumber\\
=&\frac{\dot{\mu}(t)}{\mu(t)}\mu(t)p^T(t)p(t)+2\mu(t)p^T(t)[c_1p(t)+c_2p(t-\Pi(t))\nonumber\\
&-\mathrm{diag}(\mathrm{sgn}(p(t)))(c_3{\bf 1}+c_4|p(t)|)]\nonumber\\
\le&\big[\frac{\dot{\mu}(t)}{\mu(t)}+2c_1\big]V_1(t)+|c_2|\big[\varepsilon_1\mu(t)p^T(t)p(t)\nonumber\\
&+\varepsilon_1^{-1}\mu(t)\sum_{i=1}^mp_i(t-\pi_i(t))^2\big]\nonumber\\
&-2\mu(t)[c_3|p(t)|^T\cdot{\bf 1}+c_4p(t)^Tp(t)]\nonumber\\
\le&\big[\frac{\dot{\mu}(t)}{\mu(t)}+2(c_1-c_4)+|c_2|\varepsilon_1\big]V_1(t)+|c_2|\varepsilon_1^{-1}\frac{\mu(t)}{\mu(t-\pi(t))}\nonumber\\
&\cdot\sum_{i=1}^m\frac{\mu(t-\pi(t))}{\mu(t-\pi_i(t))}\mu(t-\pi_i(t))p_i(t-\pi_i(t))^2\label{mafan1}\\
\le&\big[\frac{\dot{\mu}(t)}{\mu(t)}+2(c_1-c_4)+|c_2|\varepsilon_1\big]V_1(t)+|c_2|\varepsilon_1^{-1}\frac{\mu(t)}{\mu(t-\pi(t))}\nonumber\\
&\cdot\sum_{i=1}^m\mu(t-\pi_i(t))\sum_{k=1}^mp_k(t-\pi_i(t))^2\nonumber\\
\le&\big[\frac{\dot{\mu}(t)}{\mu(t)}+2(c_1-c_4)+|c_2|\varepsilon_1\big]V_1(t)\nonumber\\
&+|c_2|\varepsilon_1^{-1}\frac{\mu(t)}{\mu(t-\pi(t))}\sum_{i=1}^mV_1(t-\pi_i(t))\nonumber\\
\le&\big[\frac{\dot{\mu}(t)}{\mu(t)}+2(c_1-c_4)+|c_2|\varepsilon_1\big]V_1(t)\nonumber\\
&+|c_2|m\varepsilon_1^{-1}\frac{\mu(t)}{\mu(t-\pi(t))}V_1(t)\label{mafan2}\\
\le&\big[\frac{\dot{\mu}(t)}{\mu(t)}+2(c_1-c_4)+|c_2|\varepsilon_1+|c_2|m\varepsilon_1^{-1}\frac{\mu(t)}{\mu(t-\pi(t))}\big]V_1(t)\nonumber\\
<&0\nonumber
\end{align}

Therefore,
\begin{align}
W_1(T)&\ge W_1(t)=\sup_{t-\pi(t)\le s\le t}\mu(s)p(s)^Tp(s)\nonumber\\
&\ge \mu(t-\pi(t))\sup_{t-\pi(t)\le s\le t}p(s)^Tp(s),~ t\ge T\nonumber
\end{align}
i.e.,
\begin{align}
\sup_{t-\pi(t)\le s\le t}p(s)^Tp(s)\le \frac{W_1(T)}{\mu(t-\pi(t))},
\end{align}

Since $\lim_{t\rightarrow +\infty}\mu(t-\pi(t))=+\infty$, there exists $T_1(\ge T)$,
\begin{align}\label{2norm}
\sup_{t-\pi(t)\le s\le t}\|p(s)\|_2\le 1,~~ t\ge T_1
\end{align}
i.e., system (\ref{const}) with initial values would across the hyperplane with $\sup_{t-\pi(t)\le s\le t}p(s)^Tp(t)=1$ in $T_1$.

{\bf Phase II:} Next, we will prove that the norm of $p(t)$ will decrease from $1$ to $0$ in FnT.

From condition (\ref{c2}), we can find a sufficiently small constant $\varepsilon_2>0$, such that $|c_2|-c_3+\varepsilon_2<0$. Using this $\varepsilon_2$, we define another Lyapunov function
\begin{align}\label{v2}
V_2(t)=\big(p(t)^Tp(t)\big)^{\frac{1}{2}}+\varepsilon_2t,~~ t\ge T_1
\end{align}
and the maximum-value function
\begin{align}\label{w2}
W_2(t)=\sup_{t-\pi(t)\le s\le t}V_2(s),~~ t\ge T_1
\end{align}

Similarly, $V_2(t)\le W_2(t)$. If $V_2(t)<W_2(t)$, then there must exist a constant $\delta_2>0$ such that $V_2(s)\le W_2(t)$ for $s\in (t,t+\delta_2)$, i.e., $W_2(s)=W_2(t)$ for $s\in (t,t+\delta_2)$.
	
Else if there exists a point $t_2\ge T_1$, $V_2(t_2)= W_2(t_2)$. Differentiating $V_2(t)$ along (\ref{const}),
\begin{align}
&\dot{V}_2(t){|}_{t=t_2}=\big(p(t)^Tp(t)\big)^{-\frac{1}{2}}p^T(t)[c_1p(t)+c_2p(t-\Pi(t))\nonumber\\
&-\mathrm{diag}(\mathrm{sgn}(p(t)))(c_3{\bf 1}+c_4|p(t)|)]+\varepsilon_2\nonumber\\
=&\big(p(t)^Tp(t)\big)^{-\frac{1}{2}}\big[(c_1-c_4)p(t)^Tp(t)\nonumber\\
&+|c_2|\sum_{i=1}^m|p_i(t)||p_i(t-\pi_i(t))|-c_3\sum_{i=1}^m|p_i(t)|\big]+\varepsilon_2\nonumber\\
\le&\big(p(t)^Tp(t)\big)^{-\frac{1}{2}}\big[(|c_2|-c_3)\sum_{i=1}^m|p_i(t)|\big]+\varepsilon_2\nonumber\\
\le&(|c_2|-c_3)\big(p(t)^Tp(t)\big)^{-\frac{1}{2}}\big(\sum_{i=1}^mp_i(t)^2\big)^{\frac{1}{2}}+\varepsilon_2\label{mafan3}\\
=&|c_2|-c_3+\varepsilon_2<0\nonumber
\end{align}
where (\ref{mafan3}) is obtained due to $\|p(t)\|_2\le\|p(t)\|_{1}$.

Therefore, for all $t\ge T_1$
\begin{align*}
&\big(p(t)^Tp(t)\big)^{\frac{1}{2}}+\varepsilon_2t=V_2(t)\le W_2(t)\le W_2(T_1) \\
=&\sup_{T_1-\pi(T_1)\le s\le T_1}\big(\big(p(s)^Tp(s)\big)^{\frac{1}{2}}+\varepsilon_2 s\big)\\
\le&\sup_{T_1-\pi(T_1)\le s\le T_1}\big(p(s)^Tp(s)\big)^{\frac{1}{2}}+\varepsilon_2 T_1=1+\varepsilon_2 T_1
\end{align*}
i.e.,
\begin{align}\label{key}
\big(p(t)^Tp(t)\big)^{\frac{1}{2}}\le 1-\varepsilon_2(t-T_1)
\end{align}
According to the above inequality, norm of $p(t)$ will evolve from $1$ to $0$ in finite (fixed) time. Especially, $p(t)\equiv 0$ for any
\begin{align}\label{fixed}
t\ge T_2=\frac{1}{\varepsilon_2}+T_1.
\end{align}
The proof is completed.
\end{proof}

\begin{rem}
The advantage of the introduced maximum-value functions $W_1(t)$ in (\ref{w1}) and $W_2(t)$ in (\ref{w2}) is that it can relax the requirement of time delays, i.e., time delays can be asynchronous. When all time delays are the same, i.e., $\pi_i(t)=\pi(t)$, then the process from (\ref{mafan1}) to (\ref{mafan2}) can be improved with `$m$' in (\ref{mafan2}) being replaced by `$1$'.
\end{rem}

\begin{rem}
According to the form of external control, we can find that $|p(t)|^0, |p(t)|^1$ are both needed. In fact, $|p(t)|^{\alpha}, 0<\alpha<1$ can also be added, but it is not necessary. Related works using this term to realize FnT anti-synchronization can be found in \cite{wangchen2018}-\cite{liuli2020}. If FxTSta is required, then the term $|p(t)|^{\alpha}$, where $\alpha>1$ should be added, its proof is similar to the above, the first phase will use FxT, since the second phase is FxT, so the whole time will be FxT.
\end{rem}

\begin{rem}
In fact, Phase II can be proved by using an alternative function as $V_2(t)=p(t)^Tp(t)$ and differentiating it directly, interested reader can do this by yourself. In the proof, we do not adopt this function due to our proposed function (\ref{v2}) can be applied in other cases such as FnTSta under adaptive rule, which will be shown later.
\end{rem}

Next, we consider some concrete forms of time delays.
\begin{col}\label{proportional}
(FnTSta with proportional delay)
For system (\ref{const}), with $\pi_i(t)\le qt$, if there exists a positive constant $\varepsilon_1$, such that (\ref{c2}) and
\begin{align}
2(c_1-c_4)+|c_2|\varepsilon_1+|c_2|m\varepsilon_1^{-1}(1-q)^{-\varrho}<0,\label{power1}
\end{align}
hold, where $\varrho>0$ is a small positive constant, then FnTSta with proportional delay for (\ref{const}) will be realized.
\end{col}
\begin{proof}
Since $\pi(t)=qt, 0<q<1$, we can choose $\mu(t)=t^{\varrho}, \varrho>0$. Obviously, $\mu(t)$ is increasing,
\begin{align*}
&\lim_{t\rightarrow+\infty}\frac{\dot{\mu}(t)}{\mu(t)}=\lim_{t\rightarrow+\infty}\frac{\varrho t^{\varrho-1}}{t^{\varrho}}=\lim_{t\rightarrow+\infty}\frac{\varrho}{t}=0\\
&\lim_{t\rightarrow+\infty}\frac{\mu(t)}{\mu(t-\pi(t))}=\lim_{t\rightarrow+\infty}\frac{t^{\varrho}}{(t-qt)^{\varrho}}=(1-q)^{-\varrho}
\end{align*}
According to Theorem \ref{const2normthm}, the proof is completed.
\end{proof}

More types of unbounded time delays $\pi(t)$ and corresponding functions $\mu(t)$ can be found in pioneering work \cite{chenmu}.

\begin{col}\label{constantdelay}
(FnTSta with bounded time delay)
For system (\ref{const}), with $\pi_i(t)\le\pi$, if there exists a positive constant $\varepsilon_1$, such that (\ref{c2}) and
\begin{align}
&\varpi+2(c_1-c_4)+|c_2|\varepsilon_1+|c_2|m\varepsilon_1^{-1}e^{\varpi\pi}<0,\label{a}
\end{align}
hold, where $\varpi>0$ is sufficiently small, then the FnTSta with bounded delays for (\ref{const}) will be realized. Especially, when $\pi=0$, then it becomes the FnTSta problem without delays, condition (\ref{a}) becomes $\varpi+2(c_1-c_4)<0$.
\end{col}
\begin{proof}
Since $\pi(t)=\pi$ is a constant, so we can choose $\mu(t)=e^{\varpi t}, \varpi>0$, which is increasing,
\begin{align*}
&\lim_{t\rightarrow+\infty}\frac{\dot{\mu}(t)}{\mu(t)}=\lim_{t\rightarrow+\infty}\frac{\varpi e^{\varpi t}}{e^{\varpi t}}=\varpi\\
&\lim_{t\rightarrow+\infty}\frac{\mu(t)}{\mu(t-\pi(t))}=\lim_{t\rightarrow+\infty}\frac{e^{\varpi t}}{e^{\varpi (t-\pi)}}=e^{\varpi\pi}
\end{align*}
According to Theorem \ref{const2normthm}, it is proved.
\end{proof}

\begin{rem}
We are interested in the role of each term in the external control. From condition (\ref{c1}), one knows that $c_4$ should be larger than $c_1$; from condition (\ref{c2}), $c_3$ should be larger than $|c_2|$, and the fixed settling time from $1$ to $0$ is determined by $1/(c_3-|c_2|)$, so larger $c_3$ means the shorter stabilization time. Therefore, one can conclude that the larger the control parameters $c_3$ and $c_4$, the faster FnTSta can be achieved.
\end{rem}

\subsection{AFnTSta under $2$-norm by using 2PM}
In some cases or environments, parameters like $c_1$ and $c_2$ may not be obtained. Therefore, we will adapt the control parameters $c_3$ and $c_4$ as large as possible, the adaptive rule will be designed and its validity will also be proved.

\begin{thm}\label{athm}
For the following system
\begin{align}\label{adaptivethm}
\dot{p}(t)=&c_1p(t)+c_2p(t-\Pi(t))\nonumber\\
&-\mathrm{diag}(\mathrm{sgn}(p(t)))(c_3(t){\bf 1}+c_4(t)|p(t)|)
\end{align}
with the adaptive rule
\begin{align}\label{dc3}
\dot{c}_3(t)=\left\{
\begin{array}{ll}
0,&~\mathrm{if}~\sup\limits_{t-\pi(t)\le s\le t}p(s)^Tp(s)>1\\
{d_1},&~\mathrm{if}~\sup\limits_{t-\pi(t)\le s\le t}p(s)^Tp(s)\le 1\\
0,&~\mathrm{if}~\sup\limits_{t-\pi(t)\le s\le t}p(s)^Tp(s)=0
\end{array}
\right.
\end{align}
and
\begin{align}\label{dc4}
\dot{c}_4(t)=\left\{
\begin{array}{ll}
d_2\mu(t)p(t)^Tp(t),&~\mathrm{if}~\sup\limits_{t-\pi(t)\le s\le t}p(s)^Tp(s)>1\\
{d_3}\|p(t)\|_2,&~\mathrm{if}~\sup\limits_{t-\pi(t)\le s\le t}p(s)^Tp(s)\le 1
\end{array}
\right.
\end{align}
where parameters $d_1, d_2, d_3$ are any positive scalars, and the initial values are set to be $c_3(0)=c_4(0)=0$. Then FnTSta can be realized.
\end{thm}

\begin{proof} 2PM will also be applied. With the same argument with that in Theorem \ref{const2normthm}, the time is considered from time $T$.

{\bf Phase I:} Prove that the norm of $p(t)$ will decrease from initial values to $1$ in FnT.

Define a new Lyapunov function
\begin{align}\label{v3}
V_3(t)=\mu(t)p(t)^Tp(t)+\frac{1}{d_2}(c_4(t)-c_4^{\star})^2
\end{align}
and the function $W_3(t)=\sup_{t-\pi(t)\le s \le t}V_3(s)$. $V_3(t)\leq W_3(t), t\ge T$.
	
If $V_3(t)<W_3(t)$. Then there must exist a constant $\delta_3>0$ such that $V_3(s)\le W_3(t)$ for $s\in (t,t+\delta_3)$, i.e., $W_3(s)=W_3(t)$ for $s\in (t,t+\delta_3)$.
	
Else, there exists a point $t_3$, $V_3(t_3)= W_3(t_3)$. Before differentiating, we discuss the limit of $c_4(t)$.

Case 1: From the definition of $\dot{c}_4(t)$ in (\ref{dc4}), we know that $c_4(t)$ is positive and monotonically increasing, if $c_4(t)>c_4^{\star}, t>T^{\star}$, where $c_4^{\star}$ is a constant such that $\beta+2(c_1-c_4^{\star})+|c_2|\varepsilon_1+|c_2|m\varepsilon_1^{-1}(1+\eta)<0$, then FnTSta can be achieved according to the analysis in Theorem \ref{const2normthm}.

Case 2: The limit of $c_4(t)$ is always bounded by $c_4^{\star}$, $(c_4(t)-c_4^{\star})^2$ would be a decreasing function, thus $\mu(t)p(t)^Tp(t)$ must be the maximum value in $[t-\pi(t), t]$.

Now, based on the above discussions, differentiating $V_3(t)$,
\begin{align}
&\dot{V}_3(t){|}_{t=t_3}\nonumber\\
=&\dot{\mu}(t)p(t)^Tp(t)+2\mu(t)p(t)^T\dot{p}(t)+{2}{d_2^{-1}}(c_4(t)-c_4^{\star})\dot{c}_4(t)\nonumber\\
=&\frac{\dot{\mu}(t)}{\mu(t)}\mu(t)p^T(t)p(t)+2\mu(t)p^T(t)[c_1p(t)+c_2p(t-\Pi(t))\nonumber\\
&-c_4(t)p(t)]+2(c_4(t)-c_4^{\star})\mu(t)p(t)^Tp(t)\nonumber\\
\le&\big[\frac{\dot{\mu}(t)}{\mu(t)}+2(c_1-c_4^{\star})+|c_2|\varepsilon_1+|c_2|\varepsilon_1^{-1}\frac{{\mu}(t)}{\mu(t-\pi(t))}\big]V_1(t)\nonumber\\
<&0\nonumber
\end{align}

Therefore,
\begin{align}
W_3(T)&\ge W_3(t)\ge\sup_{t-\pi(t)\le s\le t}\mu(s)p(s)^Tp(s)\nonumber\\
&\ge \mu(t-\pi(t))\sup_{t-\pi(t)\le s\le t}p(s)^Tp(s),\nonumber
\end{align}
i.e.,
\begin{align}
\sup_{t-\pi(t)\le s\le t}p(s)^Tp(s)\le \frac{W_3(T)}{\mu(t-\pi(t))}.
\end{align}
So there is a FnT point $T_3(\ge T)$ such that
\begin{align}\label{2normad}
\sup_{t-\pi(t)\le s\le t}p(s)^Tp(s)\le 1,~~t\ge T_3
\end{align}

{\bf Phase II:} Prove that the norm of $p(t)$ will decrease from $1$ to $0$ in FnT.

From condition (\ref{c2}), one can find a sufficiently large constant $c_3^{\star}>0$ and sufficiently small constant $\varepsilon_2^{\star}>0$, such that $\sqrt{m}|c_2|-c_3^{\star}+\varepsilon_2^{\star}<0$. Using this $\varepsilon_2^{\star}$, we define
\begin{align}\label{v4}
V_4(t)=&\big(p(t)^Tp(t)\big)^{\frac{1}{2}}+\frac{1}{2d_1}(c_3(t)-c_3^{\star})^2\nonumber\\
&+\frac{1}{2d_3}(c_4(t)-c_4^{\star})^2+\varepsilon_2^{\star}t,~~ t\ge T_3
\end{align}
and $W_4(t)=\sup_{t-\pi(t)\le s\le t}V_4(s), t\ge T_3$. Similarly, $V_4(t)\le W_4(t)$. If $V_4(t)<W_4(t)$, then there must exist a constant $\delta_4>0$ such that $V_4(s)\le W_4(t)$ for $s\in (t,t+\delta_4)$, i.e., $W_4(s)=W_4(t)$ for $s\in (t,t+\delta_4)$.
	
Else, there exists a time point $t_4\ge T_3$ such that $V_4(t_4)= W_4(t_4)$. Similar discussions can be applied for $c_3(t)$ as before, here we omit them. Differentiating $V_4(t)$, we have
\begin{align}
&\dot{V}_4(t){|}_{t=t_4}=\big(p(t)^Tp(t)\big)^{-\frac{1}{2}}p^T(t)[c_1p(t)+c_2p(t-\Pi(t))\nonumber\\
&-\mathrm{diag}(\mathrm{sgn}(p(t)))(c_3(t){\bf 1}+c_4(t)|p(t)|)]+\varepsilon_2^{\star}\nonumber\\
&+(c_3(t)-c_3^{\star})+(c_4(t)-c_4^{\star})\|p(t)\|_2\nonumber\\
=&\big(p(t)^Tp(t)\big)^{-\frac{1}{2}}\big[(c_1-c_4^{\star})p(t)^Tp(t)\nonumber\\
&+|c_2|\sum_{i=1}^m|p_i(t)||p_i(t-\pi_i(t))|-c_3(t)\sum_{i=1}^m|p_i(t)|\big]\nonumber\\
&+(c_3(t)-c_3^{\star})+\varepsilon_2^{\star}\nonumber\\
\le&|c_2|\frac{\|p(t)\|_1}{\|p(t)\|_2}-c_3(t)\frac{\|p(t)\|_1}{\|p(t)\|_2}+(c_3(t)-c_3^{\star})+\varepsilon_2^{\star}\nonumber\\
\le&\sqrt{m}|c_2|-c_3^{\star}+\varepsilon_2^{\star}<0\nonumber
\end{align}
where the last inequality is obtained due to the norm equivalence $\|p(t)\|_2\le\|p(t)\|_{1}\le\sqrt{m}\|p(t)\|_2$.

Therefore,
\begin{align*}
&\big(p(t)^Tp(t)\big)^{\frac{1}{2}}+\varepsilon_2^{\star}t\le V_4(t)\le W_4(t)\le W_4(T_3) \\
\le&1+\varepsilon_2^{\star} T_3\\
&+\frac{(c_3(T_3-\pi(T_3))-c_3^{\star}))^2}{2d_1}+\frac{(c_4(T_3-\pi(T_3))-c_4^{\star})^2}{2d_3}\\
\le&1+\varepsilon_2^{\star} T_3+\frac{(c_3^{\star})^2}{2d_1}+\frac{(c_4^{\star})^2}{2d_3},
\end{align*}
i.e.,
\begin{align}
\big(p(t)^Tp(t)\big)^{\frac{1}{2}}\le 1+\frac{(c_3^{\star})^2}{2d_1}+\frac{(c_4^{\star})^2}{2d_3}-\varepsilon_2^{\star}(t-T_3)
\end{align}
	
From the above inequality, norm of $p(t)$ will evolve from $1$ to $0$ in FnT. Especially, $p(t)\equiv 0$ for any
\begin{align}
t\ge T_4=(\varepsilon_2^{\star})^{-1}\big[1+\frac{(c_3^{\star})^2}{2d_1}+\frac{(c_4^{\star})^2}{2d_3}\big]+T_3
\end{align}
The proof is completed.
\end{proof}

\subsection{FnTSta under other norms by using 2PM}
Besides $2$-norm, there are other norms, such as $1$-norm, $\infty$-norm, etc. Next, we present the validity of proposed 2PM on FnTSta with $1$-norm and $\infty$-norm.
\begin{thm}\label{1normthm}
FnTSta with unbounded time delays for (\ref{const}) will be realized, if condition (\ref{c2}) and the following condition hold:
\begin{align*}
\beta+(c_1-c_4)+|c_2|m(1+\eta)<0.
\end{align*}
\end{thm}

\begin{thm}\label{athmnorm1}
For (\ref{adaptivethm}) with the adaptive rules (\ref{dc3}) and
\begin{align*}
\dot{c}_4(t)=\left\{
\begin{array}{ll}
d_2\mu(t)\|p(t)\|_1,&~\mathrm{if}~\sup\limits_{t-\pi(t)\le s\le t}\|p(s)\|_1>1\\
{d_3}\|p(t)\|_1,&~\mathrm{if}~\sup\limits_{t-\pi(t)\le s\le t}\|p(s)\|_1\le 1
\end{array}
\right.
\end{align*}
FnTSta can be realized.
\end{thm}

\begin{thm}\label{wqnormthm}
FnTSta with unbounded time delays for (\ref{const}) will be realized, if condition (\ref{c2}) and the following condition hold:
\begin{align*}
\beta+(c_1-c_4)+|c_2|(1+\eta)<0.
\end{align*}
\end{thm}

\begin{thm}\label{athmnormwq}
For (\ref{adaptivethm}) with the adaptive rules (\ref{dc3}) and
\begin{align*}
\dot{c}_4(t)=\left\{
\begin{array}{ll}
d_2\mu(t)\|p(t)\|_{\infty},&~\mathrm{if}~\sup\limits_{t-\pi(t)\le s\le t}\|p(s)\|_{\infty}>1\\
{d_3}\|p(t)\|_{\infty},&~\mathrm{if}~\sup\limits_{t-\pi(t)\le s\le t}\|p(s)\|_{\infty}\le 1
\end{array}
\right.
\end{align*}
FnTSta can be realized.
\end{thm}

All proofs can be found in the Appendixes.


\section{Finite time outer synchronization of complex networks with delays}\label{FnTSyn}
As an application of obtained FnTSta results, we will investigate FnTOSyn for delayed networks.
\subsection{Model description}
We adopt the simple drive-response coupled model, while the drive system with $N$ nodes can be described as
\begin{align}
\dot{x}_i(t)=&f(x_i(t))+\theta_1\sum_{j=1}^{N}a_{ij} x_j(t)\nonumber\\
&+\theta_2\sum_{j=1}^{N}b_{ij} g(x_j({t-\pi_{ij}(t))}) \label{masterSys}
\end{align}
where $x_i(t)=(x_{i1}(t),\cdots,x_{in}(t))^T\in R^n, i=1,\cdots,N$, $f(x_i(t))=(f_1(x_i(t)),\cdots,f_n(x_i(t)))^T$, $g(x_{j}(t-\pi_{ij}(t))=(g_1(x_{j}(t-\pi_{ij}(t)),\cdots,g_n(x_{j}(t-\pi_{ij}(t)))^T$. Asynchronous time delay $0\le \pi_{ij}(t)\le \pi(t)$ satisfies Assumption \ref{as1}.

The response system can be defined as
\begin{align}
\dot{y}_i(t)=&f(y_i(t))+\theta_1\sum_{j=1}^{N}a_{ij} y_j(t)\nonumber\\
&+\theta_2\sum_{j=1}^{N}b_{ij} g(y_j({t-\pi_{ij}(t))})+u_i(t) \label{slaveSys}
\end{align}
where $y_i(t)=(y_{i1}(t),\cdots,y_{in}(t))^T$, $u_i(t)$ is external control.

\begin{asu}\label{lip}
As for functions $f(\cdot)$ and $g(\cdot)$, suppose there exist positive constants $L_f, L_g$ such that, for any $x_i, y_i\in R^n$,
\begin{align*}
&\|f(y_i)-f(x_i)\|_2\le L_f\|y_i-x_i\|_2,\\
&\|g(y_i)-g(x_i)\|_2\le L_g\|y_i-x_i\|_{2}.
\end{align*}
\end{asu}

\begin{asu}\label{matrix}
As for the irreducible coupling matrix $A=(a_{ij})$, suppose it is a Metzler matrix with zero row sum.
\end{asu}

\begin{lem}\label{left}
(\cite{Chenliulu2007})
If matrix $A$ satisfies Assumption \ref{matrix}, then the new matrix $\tilde{A}=A-\mathrm{diag}(\sigma,0,\cdots,0)$ is Lyapunov stable, $\sigma>0$. Denote $\xi=(\xi_1,\cdots,\xi_N)^T$ is the left eigenvector corresponding to the zero eigenvalue of $A$, then $\xi_i>0, i=1,\cdots,N$. In the following, we always suppose that $\sum_{i=1}^N\xi_i=1$. Moreover, $\Xi\tilde{A}+\tilde{A}\Xi=2\{\Xi\tilde{A}\}^s<0$, where $\Xi=\mathrm{diag}(\xi)$, i.e., its largest eigenvalue $\lambda_{\max}(\{\Xi\tilde{A}\}^s)<0$.
\end{lem}

\begin{mydef}
Systems (\ref{masterSys}) and (\ref{slaveSys}) are said to achieve the FnTOSyn, if for the synchronization error $e_i(t)=y_i(t)-x_i(t), i=1,\cdots N$, there exists a time point $T^{\star}$ such that
\begin{align*}
\lim\limits_{t\rightarrow T^{\star}}\|e_i(t)\|=0, ~~\mathrm{and}~~ \|e_i(t)\|\equiv 0,~ t\ge T^{\star}.
\end{align*}
\end{mydef}

Denote $\tilde{f}(e_i(t))=f(y_i(t))-f(x_i(t))$ and $\tilde{g}(e_i(t))=g(y_i(t))-g(x_i(t))$, then the dynamic model for error $e_i(t)$ can be described as:
\begin{align}\label{errorSys}
\dot{e}_i(t)=&\tilde{f}(e_i(t))+\theta_1\sum_{j=1}^{N}a_{ij} e_j(t)\nonumber\\
&+\theta_2\sum_{j=1}^{N}b_{ij}\tilde{g}(e_j(t-\pi_{ij}(t)))+u_i
\end{align}

Now, our aim is to design suitable external control $u_i(t)$ and prove its validity to realize FnTOSyn. Of course, the external control should be simple and easy to use. Under this standard, $u_i(t)$ should be only dependent on the current system state $e_i(t)$; moreover, pinning control strategy should also be considered in order to reduce control cost.

\subsection{FnTOSyn using 2PM}
In this subsection, we first design the form of $u_i(t)$ by the pinning control technique, then FnTOSyn under this control will also be proved using the proposed 2PM.

Design the controller $u_i(t)$ as
\begin{align}\label{control}
u_i(t)=\left\{
\begin{array}{ll}
-\theta_1\sigma e_1(t)-\theta_3\mathrm{sgn}(e_{1}(t)), &i=1\\
-\theta_3\mathrm{sgn}(e_{i}(t)), &\mathrm{otherwise}
\end{array}
\right.
\end{align}
where $\sigma>0$ and $\theta_3>0$.

\begin{thm}\label{first}
For the system (\ref{errorSys}) with control (\ref{control}), if there exists a constant $\varepsilon_1>0$, such that following conditions hold:
\begin{align}\label{complexc1}
&\beta+2L_f+\theta_2\max_{i,j}|b_{ij}|N\varepsilon_1+2\theta_1\lambda_{\max}(\{\Xi\tilde{A}\}^s)\nonumber\\
&~~+\theta_2\max_{i,j}|b_{ij}|N^2nL_g^2\varepsilon_1^{-1}\frac{1}{\min_i\xi_i}(1+\eta)<0
\end{align}
and
\begin{align}\label{complexc2}
\theta_2\max_{ij}|b_{ij}|NL_g-\theta_3<0,
\end{align}
then (\ref{errorSys}) will achieve FnTOSyn.
\end{thm}

\begin{proof}
With the same argument as that in Theorem \ref{const2normthm}, the time is considered from time $T$, and according to condition (\ref{complexc1}), we have
\begin{align*}
&\frac{\dot{\mu}(t)}{\mu(t)}+2L_f+\theta_2\max_{i,j}|b_{ij}|N\varepsilon_1+2\theta_1\lambda_{\max}(\{\Xi\tilde{A}\}^s)\nonumber\\
&+\theta_2\max_{i,j}|b_{ij}|N^2nL_g^2\varepsilon_1^{-1}\frac{1}{\min_i\xi_i}\frac{{\mu(t)}}{\mu(t-\pi(t))}<0.
\end{align*}
If $\sup_{t-\pi(t)\le s\le t}\sum_{i=1}^{N}e_i^T(s)e_i(s)\le 1, t\ge T$, then we will switch to Phase II directly; else, we will start from Phase I.

{\bf Phase I}: Define
\begin{align}\label{ov1}
\overline{V}_1(t)=\mu(t)\sum_{i=1}^{N}\xi_{i}e_i(t)^Te_i(t),~~ t\ge T.
\end{align}
and
\begin{align}\label{ow1}
\overline{W}_1(t)=\sup_{t-\pi(t)\le s \le t}\overline{V}_1(s)
\end{align}
If at some time $t_1$, $\overline{V}_1(t_1)=\overline{W}_1(t_1)$, differentiating $\overline{V}_1(t)$ along (\ref{errorSys}),
\begin{align}
&\dot{\overline{V}}_1(t)=\dot{\mu}(t)\sum_{i=1}^{N}\xi_{i}e_i(t)^Te_i(t)\nonumber\\
&+2\mu(t)\sum_{i=1}^{N}\xi_{i}e_i(t)^T[\tilde{f}(e_i(t))+\theta_1\sum_{j=1}^{N}a_{ij} e_j(t)\nonumber\\
&+\theta_2\sum_{j=1}^{N}b_{ij}\tilde{g}(e_j(t-\pi_{ij}(t)))+u_i]\nonumber\\
\le&\frac{\dot{\mu}(t)}{\mu(t)}\overline{V}_1(t)+2L_f\mu(t)\sum_{i=1}^{N}\xi_{i}e_i(t)^Te_i(t)\nonumber\\
&+2\mu(t)\theta_1E(t)^T(\{\Xi\tilde{A}\}^s\otimes I_n)E(t)\nonumber\\
&+2\mu(t)\theta_2\sum_{i=1}^N\sum_{j=1}^N\xi_ie_i(t)^Tb_{ij}\tilde{g}(e_j(t-\pi_{ij}(t)))\nonumber\\
\le&\big(\frac{\dot{\mu}(t)}{\mu(t)}+2L_f+2\theta_1\lambda_{\max}(\{\Xi\tilde{A}\}^s)\big)\overline{V}_1(t)\nonumber\\
&+\mu(t)\theta_2\max_{i,j}|b_{ij}|\sum_{i=1}^N\sum_{j=1}^N\big[\varepsilon_1\xi_ie_i(t)^Te_i(t)\nonumber\\
&+\varepsilon_1^{-1}\xi_i\tilde{g}(e_j(t-\pi_{ij}(t)))^T\tilde{g}(e_j(t-\pi_{ij}(t)))\big]\nonumber\\
\le&\big[\frac{\dot{\mu}(t)}{\mu(t)}+2L_f+\theta_2\max_{i,j}|b_{ij}|N\varepsilon_1+2\theta_1\lambda_{\max}(\{\Xi\tilde{A}\}^s)\big]\overline{V}_1(t)\nonumber\\
+&\theta_2\max_{i,j}|b_{ij}|NL_g^2\varepsilon_1^{-1}\frac{1}{\min_i\xi_i}{\mu(t)}\sum_{i,j}\overline{V}_1(t-\pi_{ij}(t))\nonumber\\
\le&\bigg[\frac{\dot{\mu}(t)}{\mu(t)}+2L_f+\theta_2\max_{i,j}|b_{ij}|N\varepsilon_1+2\theta_1\lambda_{\max}(\{\Xi\tilde{A}\}^s)\nonumber\\
&+\theta_2\max_{i,j}|b_{ij}|N^2nL_g^2\varepsilon_1^{-1}\frac{1}{\min_i\xi_i}\frac{{\mu(t)}}{\mu(t-\pi(t))}\bigg]\overline{V}_1(t)\nonumber\\
<&0\nonumber
\end{align}
where $E(t)=(e_1(t)^T, \cdots, e_N(t)^T)^T$.

According to the same arguments in Theorem \ref{const2normthm}, system (\ref{errorSys}) would decrease and across the hyperplane with $\sup_{t-\pi(t) \leq s \leq t}\sum_{i=1}^{N}e_i(t)^Te_i(t)=1$ in FnT $T_1(\ge T)$.

{\bf Phase II}: From condition (\ref{complexc2}), we can find a sufficiently small constant $\varepsilon_2>0$, such that $(\theta_2\max_{ij}|b_{ij}|Nd_2-\theta_3)+\varepsilon_2<0$. Using this $\varepsilon_2$, we define another Lyapunov function
\begin{align}\label{ov2}
\overline{V}_2(t)=\big(\sum_{i=1}^{N}\xi_{i}e_i(t)^Te_i(t)\big)^{\frac{1}{2}}+\varepsilon_2t,~~ t\ge T_1
\end{align}
and the maximum-value function
\begin{align}\label{ow2}
\overline{W}_2(t)=\sup_{t-\pi(t)\le s\le t}\overline{V}_2(s),~~ t\ge T_1
\end{align}
If at some time $t_2$, $\overline{V}_2(t_2)=\overline{W}_2(t_2)$, differentiating $\overline{V}_2(t)$,
\begin{align}
&\dot{\overline{V}}_2(t)=\big(\sum_{i=1}^{N}\xi_{i}e_i(t)^Te_i(t)\big)^{-\frac{1}{2}}\sum_{i=1}^{N}\xi_{i}e_i(t)^T\dot{e}_i(t)+\varepsilon_2\nonumber\\
\le&\varepsilon_2+\big(\sum_{i=1}^{N}\xi_{i}e_i(t)^Te_i(t)\big)^{-\frac{1}{2}}\cdot\nonumber\\
&\bigg[(L_f+\theta_1\lambda_{\max}(\{\Xi\tilde{A}\}^s))\sum_{i=1}^{N}\xi_{i}e_i(t)^Te_i(t)\nonumber\\
&+\theta_2\sum_{i=1}^N\sum_{j=1}^N\xi_ie_i(t)^Tb_{ij}\tilde{g}(e_j(t-\pi_{ij}(t)))\nonumber\\
&-\theta_3\sum_{i=1}^N\xi_ie_i(t)^T\mathrm{sgn}(e_i(t))\bigg]\nonumber\\
\le&\varepsilon_2+\big(\sum_{i=1}^{N}\xi_{i}e_i(t)^Te_i(t)\big)^{-\frac{1}{2}}\cdot\nonumber\\
&(\theta_2\max_{ij}|b_{ij}|NL_g-\theta_3)\sum_{i=1}^N\sqrt{\xi_i}|e_i(t)|^T{\bf 1}\nonumber\\
\le&\varepsilon_2+(\theta_2\max_{ij}|b_{ij}|NL_g-\theta_3)<0\nonumber
\end{align}

With the similar arguments in Theorem \ref{const2normthm}, the error $\big(\sum_{i=1}^{N}\xi_{i}e_i(t)^Te_i(t)\big)^{\frac{1}{2}}$ will flow from $1$ to $0$ in FnT. Especially, $e_i(t)\equiv 0$ for any $t\ge T_2={1}/{\varepsilon_2}+T_1$.
\end{proof}

\begin{rem}
From the above theorem, we can find that the enlargement of inequalities can be improved by discussing more carefully. For example, the property of matrix $B=(b_{ij})$ (Metzler or not), the type of time delays $\pi_{ij}(t)$ (synchronous or asynchronous, bounded or unbounded), the definition of function $g(\cdot)$ (linear or nonlinear), and even other different norms of $e_i(t)$ ($1$-norm or $\infty$-norm), etc. All these can be done following our proposed 2PM, which has a wide scope of applications for FnT analysis of delayed systems.
\end{rem}

Next, we design another controller added on each node,
\begin{align}\label{control2}
{u}_i(t)=
-\theta_3\mathrm{sgn}(e_{i}(t))-\theta_4e_i(t),
\end{align}
where $\theta_3>0$ and $\theta_4>0$, and $i=1,2,\cdots,N$.

\begin{thm}\label{second}
For the system (\ref{errorSys}) with control (\ref{control2}), if there exists a constant $\varepsilon_1>0$, such that (\ref{complexc2}) and
\begin{align}\label{complexc3}
&\beta+2L_f+\theta_2\max_{i,j}|b_{ij}|N\varepsilon_1+2\lambda_{\max}(\{\Xi\hat{A}\}^s)\nonumber\\
&~~+\theta_2\max_{i,j}|b_{ij}|N^2nL_g^2\varepsilon_1^{-1}\frac{1}{\min_i\xi_i}(1+\eta)<0,
\end{align}
hold, where $\hat{A}=\theta_1A+\theta_4I$, then (\ref{errorSys}) will achieve FnTOSyn.
\end{thm}

Its proof is similar to that in Theorem \ref{first}.

Next, we discuss classical type of synchronization, i.e., states of each nodes in the network are the same, here we call it finite time inner synchronization (FnTISyn), which is also the FnTSyn in the introduction.
\subsection{FnTISyn using 2PM}
In this case, synchronization trajectory $\phi(t)$ is a solution of an isolated node of (\ref{masterSys}),
\begin{align}
\dot{\phi}(t)=f(\phi(t)), \label{innerSys}
\end{align}
where $\phi(t)$ can be a fixed point, a chaotic orbit, or others.

If $x_1(t)=\cdots=x_N(t)=\phi(t)$ is a solution for network (\ref{masterSys}), then the following condition should hold:
\begin{align}\label{should}
\sum_{i=1}^Nb_{ij}g(\phi(t-\pi_{ij}(t)))=0
\end{align}

In the literature of classical synchronization, condition (\ref{should}) generally means the coupling relationships between node $i$ and its neighbours, so its form can be transformed as: $\sum_{i=1}^Nb_{ij}g(x_j(t-\pi_{ij}(t)))=\sum_{i=1}^Nb_{ij}[g(x_j(t-\pi_{ij}(t)))-g(x_i(t-\pi_{ij}(t)))]$. Moreover, if network (\ref{masterSys}) itself can realize FnTISyn, then the outer synchronization can also convert into inner synchronization. As for the FnTISyn for networks with delay and without control, we will address this problem in other papers, here we omit it.

Now, consider the network (\ref{slaveSys}) with external control, and define the error as $e_i(t)=y_i(t)-\phi(t)$, then its dynamics can be described as
\begin{align}
\dot{e}_i(t)&=
f(y_i(t))-f(\phi(t))+\theta_1\sum_{j=1}^{N}a_{ij} e_j(t)\label{inner}\\
+&\theta_2\sum_{j=1}^{N}b_{ij}\big[g(y_j(t-\pi_{ij}(t)))-g(\phi(t-\pi_{ij}(t)))\big]+u_i(t)\nonumber
\end{align}
where $u_i(t)$ can be defined in (\ref{control}) (or (\ref{control2})). Then, as a direct result of Theorem \ref{first} (or Theorem \ref{second}), we have
\begin{col}
If conditions (\ref{complexc1}) and (\ref{complexc2}) hold, then FnTISyn can be achieved for network (\ref{slaveSys}).
\end{col}

\subsection{AFnTOSyn using 2PM}
According to (\ref{complexc1}) and (\ref{complexc2}), one can see that the larger $\theta_1$ and $\theta_3$, the faster AFnTOSyn can be achieved. Moreover, since in some cases, we cannot obtain the important parameters, even the network configuration. Next, we will apply the adaptive strategy on them and prove the validity of designed adaptive rules.
\begin{thm}\label{complete1}
For the error system,
\begin{align}
\dot{e}_i(t)=&\tilde{f}(e_i(t))+\theta_1(t)\sum_{j=1}^{N}\tilde{a}_{ij} e_j(t)\nonumber\\
+&\theta_2\sum_{j=1}^{N}b_{ij}\tilde{g}(e_j(t-\pi_{ij}(t)))-\theta_3(t)\mathrm{sgn}(e_{i}(t))\label{adaptsyn}
\end{align}
with the following adaptive rule
\begin{align}\label{adc1}
\dot{\theta}_1(t)=\left\{
\begin{array}{l}
d_1\mu(t)\sum_{i=1}^{N}e_i(t)^Te_i(t),\\
~~~\mathrm{if}~\sup\limits_{t-\pi(t)\le s\le t}\sum_{i=1}^{N}e_i(s)^Te_i(s)>1\\
{d_2}\big(\sum_{i=1}^{N}e_i(t)^Te_i(t)\big)^{\frac{1}{2}},~~~~~\mathrm{otherwise}
\end{array}
\right.
\end{align}
and
\begin{align}\label{adc3}
\dot{\theta}_3(t)=\left\{
\begin{array}{l}
0,~\mathrm{if}~\sup\limits_{t-\pi(t)\le s\le t}\sum_{i=1}^{N}e_i(s)^Te_i(s)>1\\
0,~\mathrm{if}~\sup\limits_{t-\pi(t)\le s\le t}\sum_{i=1}^{N}e_i(s)^Te_i(s)=0\\
{d_3}, ~~~~~~~~~~~~~~~~~~~~~~~~~~~~~~~~\mathrm{otherwise}
\end{array}
\right.
\end{align}
where parameters $d_1, d_2, d_3$ are any positive scalars, and the initial values are set to be $\theta_1(0)=\theta_3(0)=0$. Then FnTOSyn can be realized.
\end{thm}

\begin{proof}
With the same argument as that in Theorem \ref{const2normthm}, the time is considered from time $T$.

{\bf Phase I:} Define a new Lyapunov function
\begin{align}\label{ov3}
\overline{V}_3(t)=&\mu(t)\sum_{i=1}^{N}\xi_{i}e_i(t)^Te_i(t)\nonumber\\
&+\frac{|\lambda_{\max}(\{\Xi\tilde{A}\}^s)|}{d_1}(\theta_1(t)-\theta_1^{\star})^2
\end{align}
and the function $\overline{W}_3(t)=\sup_{t-\pi(t)\le s \le t}\overline{V}_3(s)$.

Obviously, $\overline{V}_3(t)\le \overline{W}_3(t), t\ge T$. If there exists a point $t_3$, $\overline{V}_3(t_3)=\overline{W}_3(t_3)$. Before differentiating, we discuss the limit of $c_4(t)$.

Similarly, with the same arguments in Theorem \ref{athm}, if $\theta_1(t)>\theta_1^{\star}, t>T^{\star}$, where $\theta_1^{\star}$ is a constant such that (\ref{complexc1}) holds, then FnTOSyn can be achieved from Theorem \ref{first}. Otherwise, if the limit of $\theta_1(t)$ is always bounded by $\theta_1^{\star}$, $(\theta_1(t)-\theta_1^{\star})^2$ would be a decreasing function, thus $\mu(t)\sum_{i=1}^{N}\xi_{i}e_i(t)^Te_i(t)$ must be the maximum value in the interval $[t-\pi(t), t]$. Now, differentiating $\overline{V}_3(t)$, with the same process in Theorem \ref{first}, we have
\begin{align}
&\dot{\overline{V}}_3(t)\nonumber\\
\le&\bigg[\frac{\dot{\mu}(t)}{\mu(t)}+2L_f+\theta_2\max_{i,j}|b_{ij}|N\varepsilon_1+2\theta_1^{\star}\lambda_{\max}(\{\Xi\tilde{A}\}^s)\nonumber\\
&+\theta_2\max_{i,j}|b_{ij}|N^2nL_g^2\varepsilon_1^{-1}\frac{1}{\min_i\xi_i}\frac{{\mu(t)}}{\mu(t-\pi(t))}\bigg]\overline{V}_1(t)\nonumber\\
<&0\nonumber
\end{align}
Therefore,
\begin{align}
\overline{W}_3(T)&\ge \overline{W}_3(t)\ge\sup_{t-\pi(t)\le s\le t}\mu(s)\sum_{i=1}^{N}\xi_{i}e_i(s)^Te_i(s)\nonumber\\
&\ge \mu(t-\pi(t))\sup_{t-\pi(t)\le s\le t}\sum_{i=1}^{N}\xi_{i}e_i(s)^Te_i(s),\nonumber
\end{align}
i.e., there is a time point $T_3(\ge T)$ such that
\begin{align}
\sup_{t-\pi(t)\le s\le t}\sum_{i=1}^{N}\xi_{i}e_i(s)^Te_i(s)\le 1,~~t\ge T_3
\end{align}

{\bf Phase II:} From (\ref{complexc2}), one can find a sufficiently large constant $\theta_3^{\star}>0$ and sufficiently small constant $\varepsilon_2^{\star}>0$, such that $\theta_2\max_{ij}|b_{ij}|N^{\frac{3}{2}}n^{\frac{1}{2}}d_2-\theta_3^{\star}+\varepsilon_2^{\star}<0$. Using this $\varepsilon_2^{\star}$,
\begin{align}\label{ov4}
\overline{V}_4(t)=&\big(\sum_{i=1}^{N}\xi_{i}e_i(t)^Te_i(t)\big)^{\frac{1}{2}}+\frac{|\lambda_{\max}(\{\Xi\tilde{A}\}^s)|}{2d_2}(\theta_1(t)-\theta_1^{\star})^2\nonumber\\
&+\frac{1}{2d_3}(\theta_3(t)-\theta_3^{\star})^2+\varepsilon_2^{\star}t,~~ t\ge T_3
\end{align}
and $\overline{W}_4(t)=\sup_{t-\pi(t)\le s\le t}\overline{V}_4(s), t\ge T_3$. Similarly, $\overline{V}_4(t)\le \overline{W}_4(t)$.
If there exists a time point such that the equality holds. Similar discussions can be applied for $\theta_3(t)$ as that in $\theta_1(t)$, here we omit them. Differentiating $\overline{V}_4(t)$,
\begin{align}
&\dot{\overline{V}}_4(t)\nonumber\\
\le&\varepsilon_2^{\star}+\big(\sum_{i=1}^{N}\xi_{i}e_i(t)^Te_i(t)\big)^{-\frac{1}{2}}\cdot\nonumber\\
&\bigg[(L_f+\theta_1(t)\lambda_{\max}(\{\Xi\tilde{A}\}^s))\sum_{i=1}^{N}\xi_{i}e_i(t)^Te_i(t)\nonumber\\
&+\theta_2\sum_{i=1}^N\sum_{j=1}^N\xi_ie_i(t)^Tb_{ij}\tilde{g}(e_j(t-\pi_{ij}(t)))\nonumber\\
&-\theta_3(t)\sum_{i=1}^N\xi_ie_i(t)^T\mathrm{sgn}(e_i(t))\bigg]\nonumber\\
&+(\theta_1(t)-\theta_1^{\star})|\lambda_{\max}(\{\Xi\tilde{A}\}^s)|\big(\sum_{i=1}^{N}\xi_{i}e_i(t)^Te_i(t)\big)^{\frac{1}{2}}\nonumber\\
&+(\theta_3(t)-\theta_3^{\star})\nonumber\\
\le&\varepsilon_2^{\star}+(\theta_2\max_{ij}|b_{ij}|N^{\frac{3}{2}}n^{\frac{1}{2}}L_g-\theta_3^{\star})<0\nonumber
\end{align}

Therefore,
\begin{align*}
&\big(\sum_{i=1}^{N}\xi_{i}e_i(t)^Te_i(t)\big)^{\frac{1}{2}}+\varepsilon_2^{\star}t\le \overline{V}_4(t)\le \overline{W}_4(t)\le \overline{W}_4(T_3) \\
\le&1+\varepsilon_2^{\star} T_3+\frac{|\lambda_{\max}(\{\Xi\tilde{A}\}^s)|(\theta_1^{\star})^2}{2d_2}+\frac{(\theta_3^{\star})^2}{2d_3},
\end{align*}
i.e., $e_i(t)\equiv 0$ for any
\begin{align*}
t\ge T_4=(\varepsilon_2^{\star})^{-1}\bigg[1+\frac{|\lambda_{\max}(\{\Xi\tilde{A}\}^s)|(\theta_1^{\star})^2}{2d_2}+\frac{(\theta_3^{\star})^2}{2d_3}\bigg]+T_3
\end{align*}
The proof is completed.
\end{proof}

As a direct result of the above theorem, we have
\begin{col}
For the network (\ref{inner}) with control (\ref{control}), under the adaptive rules (\ref{adc1}) and (\ref{adc3}), FnTISyn can be realized.
\end{col}

\begin{rem}
From the adaptive rule (\ref{adc1}), one can see that the drive system (\ref{masterSys}) has been changed due to the change of $\theta_1$, therefore, if one requirement for outer synchronization is that the drive system cannot be changed, then one should use the other controller defined in (\ref{control2}).
\end{rem}

\begin{thm}
For the following error system under (\ref{control2}),
\begin{align}
&\dot{e}_i(t)=\tilde{f}(e_i(t))+\theta_1\sum_{j=1}^{N}{a}_{ij} e_j(t)\label{adaptsyn2}\\
&~+\theta_2\sum_{j=1}^{N}b_{ij}\tilde{g}(e_j(t-\pi_{ij}(t)))-\theta_3(t)\mathrm{sgn}(e_{i}(t))-\theta_4(t)e_i(t)\nonumber
\end{align}
with the following adaptive rules defined in (\ref{adc3}) and
\begin{align}\label{adc4}
\dot{\theta}_4(t)=\left\{
\begin{array}{l}
d_1\mu(t)\sum_{i=1}^{N}e_i(t)^Te_i(t),\\
~~~\mathrm{if}~\sup\limits_{t-\pi(t)\le s\le t}\sum_{i=1}^{N}e_i(s)^Te_i(s)>1\\
{d_2}\big(\sum_{i=1}^{N}e_i(t)^Te_i(t)\big)^{\frac{1}{2}},~~~~~\mathrm{otherwise}
\end{array}
\right.
\end{align}
where parameters $d_1, d_2, d_3$ are any positive scalars, and the initial values are set to be $\theta_3(0)=\theta_4(0)=0$. Then FnTOSyn can be realized.
\end{thm}

This result can be proved similarly to that in Theorem \ref{complete1}.

\section{Numerical simulations}\label{ns}
Next, we present two numerical examples to demonstrate the validity of our proposed 2PM.
\subsection{FnTSta of system with delay}
{\bf Example 1}: Consider the simplest system:
\begin{align}\label{est}
\dot{p}(t)=p(t)+2p(t-\pi(t))-c_3\mathrm{sgn}(p(t)))-c_4p(t)
\end{align}
where $p(t)\in R$.

Let $\pi(t)=0.5t$, then according to Corollary 1, $\mu(t)=t^{\varrho}$, pick $\varrho=0.1, \varepsilon_1=2^{0.05}$, then conditions (\ref{power1}) and (\ref{c2}) becomes
\begin{align}
2(1-c_4)+4\cdot2^{0.05}<0,~~\mathrm{and}~~2-c_3<0
\end{align}
so parameter $c_4$ should be larger than $1+2^{1.05}=3.071$, and $c_3$ should be larger than $2$.

Pick $c_3=2.1$, $c_4=3.5$, then Fig. \ref{conv} shows that FnTSta can be realized; Fig. \ref{compare1} shows that under the same $c_3=2.1$, the larger $c_4$, the smaller convergence time, and small $c_4$ may destroy its stability. Fig. \ref{compare2} shows that under the same $c_4=3.5$, the larger $c_3$, the smaller convergence time.
\begin{figure}
\begin{center}
\includegraphics[width=0.5\textwidth]{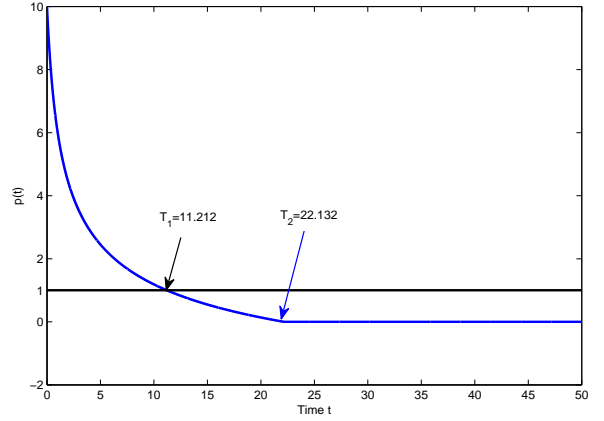}
\caption{FnTSta of system (\ref{est})}\label{conv}
\end{center}
\end{figure}
\begin{figure}
\begin{center}
\includegraphics[width=0.5\textwidth]{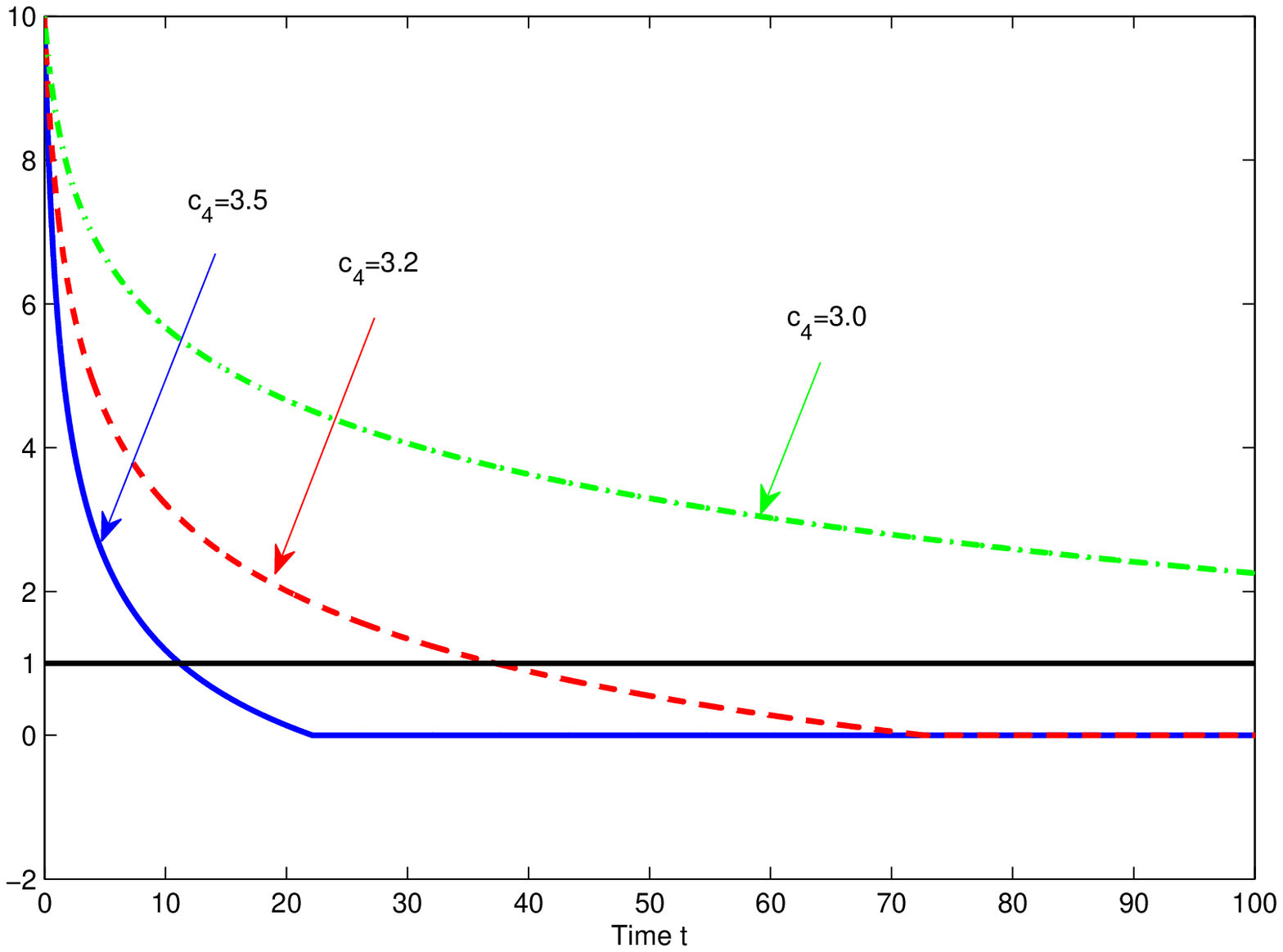}
\caption{FnTSta of system (\ref{est}) under different $c_4$}\label{compare1}
\end{center}
\end{figure}
\begin{figure}
\begin{center}
\includegraphics[width=0.5\textwidth]{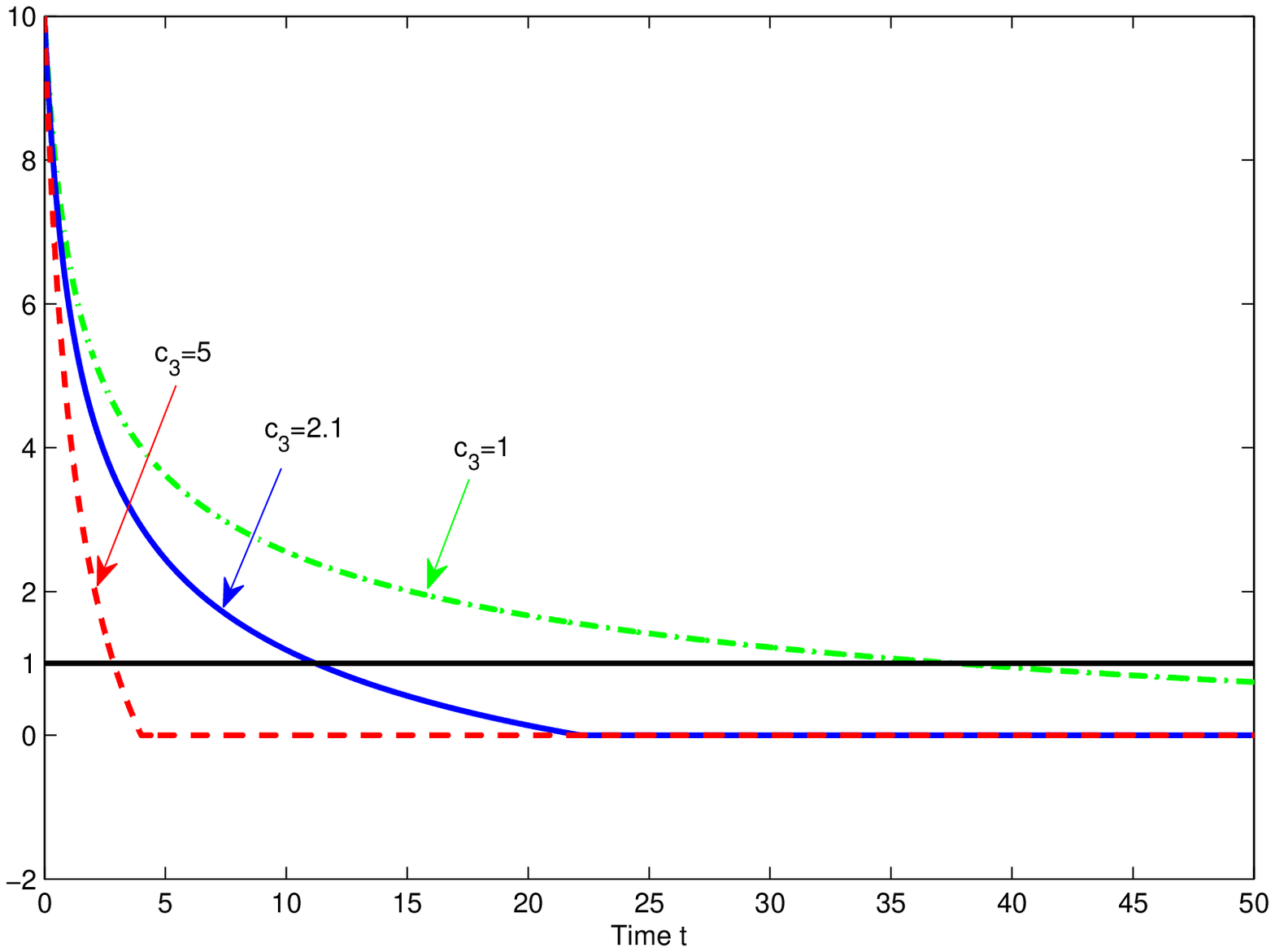}
\caption{FnTSta of system (\ref{est}) under different $c_3$}\label{compare2}
\end{center}
\end{figure}

Now, we apply the adaptive strategy to $c_3$ and $c_4$ in (\ref{est}), where the adaptive rules are defined as
\begin{align}\label{rule1}
\dot{c}_3(t)=\left\{
\begin{array}{ll}
0,&~\mathrm{if}~\sup\limits_{0.5t\le s\le t}p(s)^2>1\\
0.1,&~\mathrm{if}~\sup\limits_{0.5t\le s\le t}p(s)^2\le 1\\
0,&~\mathrm{if}~\sup\limits_{0.5t\le s\le t}p(s)^2=0
\end{array}
\right.
\end{align}
and
\begin{align}\label{rule2}
\dot{c}_4(t)=\left\{
\begin{array}{ll}
0.1t^{0.1}p(t)^2,&~\mathrm{if}~\sup\limits_{0.5t\le s\le t}p(s)^2>1\\
0.1|p(t)|,&~\mathrm{if}~\sup\limits_{0.5t\le s\le t}p(s)^2\le 1
\end{array}
\right.
\end{align}
with $c_3(0)=c_4(0)=0$, then according to Theorem \ref{athm}, AFnTSta can be realized. Fig. \ref{adeg1} shows the dynamics of $p(t)$, $c_3(t)$ and $c_4(t)$, respectively. Moreover, since the time delay $\pi(t)=0.5t$ is unbounded, and according to adaptive rules (\ref{rule1}) and (\ref{rule2}), the unchange of $c_3(t)$ and $c_4(t)$ has a long delay than the first time point $t^{\star}$ where $p(t^{\star})=0$.
\begin{figure}
\begin{center}
\includegraphics[width=0.5\textwidth]{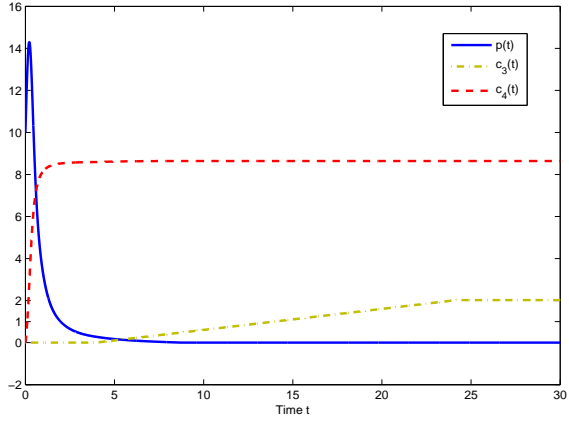}
\caption{The dynamics of $p(t)$, $c_3(t)$ and $c_4(t)$ for system (\ref{est}) under adaptive rules (\ref{rule1}) and (\ref{rule2})}\label{adeg1}
\end{center}
\end{figure}

\subsection{FnTOSyn of complex networks with delay}
{\bf Example 2}: Suppose the drive system is composed of three Lorenz oscillators,
\begin{align}\label{simuMaster}
\dot{x}_i(t)=&f(x_i(t))+0.1\sum_{j=1}^{3}a_{ij}x_j(t)+\sum_{j=1}^{3}b_{ij}g(x_j(t-\pi_{ij}(t)))
\end{align}
where $x_i(t)=(x_{i1}(t), x_{i2}(t), x_{i3}(t))^T$, the intrinsic function $f(x_i): R^3\rightarrow R^3$ is described by:
\begin{align*}
\left(
\begin{array}{c}
10(x_{i2}(t)-x_{i1}(t))\\
28{x}_{i1}(t)-{x}_{i2}(t)-x_{i1}(t){x}_{i3}(t)\\
x_{i1}(t){x}_{i2}(t)-8x_{i3}(t)/3
\end{array}\right),
\end{align*}
the coupling matrix
\begin{align*}
&A=(a_{ij})=\left(
\begin{array}{ccc}
-5 &  2 & 3\\
1 & -4  & 3\\
1 & 2 & -3\\
\end{array}
\right),\\
&B=(b_{ij})=\left(
\begin{array}{ccc}
1 &  -1 & 1\\
1 &1  & -1\\
-1 & 1 & 1\\
\end{array}
\right),
\end{align*}
the nonlinear function
\begin{align*}
g(x_i(t))=
\left(
\begin{array}{c}
\sin(x_{i1}(t))+2x_{i1}(t)\\
\sin(x_{i2}(t))+2x_{i2}(t)\\
\sin(x_{i3}(t))+2x_{i3}(t)
\end{array}
\right),
\end{align*}
and the asynchronous time delays are
\begin{align*}
\pi_{ij}(t)=0.5(1-0.1\cdot|\sin(i+2j)|)t\le \pi(t)=0.5t,
\end{align*}
and initial values are
\begin{align*}
x_1(0)=(-1.5771, 0.5080, 0.2820)^T,\\
x_2(0)=(0.0335, -1.3337, 1.1275)^T,\\
x_3(0)=(0.3502, -0.2991, 0.0229)^T.
\end{align*}
The inner error among each node in the drive system (\ref{simuMaster}) is measured by the following index:
\begin{align}\label{E1}
\|E1(t)\|=\|x_2(t)-x_1(t)\|_2+\|x_3(t)-x_1(t)\|_2
\end{align}

Define the response system as:
\begin{align}
\dot{y}_i(t)=&f(y_i(t))+0.1\sum_{j=1}^{3}a_{ij} y_j(t)+\sum_{j=1}^{3}b_{ij} g(y_j(t-\tau_{ij}(t)))\nonumber\\
&+u_i(t)
\label{simuSlave}
\end{align}
where $f(\cdot), A, B$ are defined as before. Its initial values are:
\begin{align*}
y_1(0)=&(-0.8479, -1.1201, 2.5260)^T,\\
y_2(0)=&(1.6555, 0.3075, -1.2571)^T,\\
y_3(0)=&(-0.8655, -0.1765, 0.7914)^T.
\end{align*}
The inner error among each node in the response system (\ref{simuSlave}) is measured by the following index:
\begin{align}\label{E2}
\|E2(t)\|=\|y_2(t)-y_1(t)\|_2+\|y_3(t)-y_1(t)\|_2
\end{align}
Moreover, as for the outer synchronization error between the drive system (\ref{simuMaster}) and the response system (\ref{simuSlave}), we use the following index:
\begin{align}\label{outers}
\|E(t)\|_2=\sqrt{\sum_{i=1}^3[y_i(t)-x_i(t)]^T[y_i(t)-x_i(t)]}
\end{align}

When there is no controller, i.e., $u_i(t)=0$, Fig. \ref{adeg5} (a) shows that nodes among the drive system are not synchronized, (b) shows that nodes among the response system are not synchronized, and (c) shows that outer synchronization is also not achieved.

Next, we apply the adaptive strategy to realize outer synchronization. The external controller $u_i(t)$ is defined by
\begin{align}\label{comp0}
u_i(t)=-\theta_3(t)\mathrm{sgn}(e_{i}(t))-\theta_4(t)e_i(t),~ i=1, 2, 3
\end{align}
where $e_i(t)=y_i(t)-x_i(t)$, and the adaptive rules are
\begin{align}\label{comp1}
\dot{\theta}_3(t)=\left\{
\begin{array}{l}
0.02, ~\mathrm{if}~0<\sup\limits_{t-\pi(t)\le s\le t}\sum_{i=1}^{3}e_i(s)^Te_i(s)\le 1\\
0,~~~~~~~~~~~~~~~~~~~~~~~~~~~~~~~~~~\mathrm{otherwise}
\end{array}
\right.
\end{align}
and
\begin{align}\label{comp2}
\dot{\theta}_4(t)=\left\{
\begin{array}{l}
0.05\mu(t)\sum_{i=1}^{N}e_i(t)^Te_i(t),\\
~~~\mathrm{if}~\sup\limits_{t-\pi(t)\le s\le t}\sum_{i=1}^{3}e_i(s)^Te_i(s)>1\\
0.05\big(\sum_{i=1}^{3}e_i(t)^Te_i(t)\big)^{\frac{1}{2}},~~~\mathrm{otherwise}
\end{array}
\right.
\end{align}
Fig. \ref{adeg6} shows that under the above defined adaptive controller, FnTOSyn can be finally realized.

\begin{figure}
\begin{center}
\includegraphics[width=0.5\textwidth]{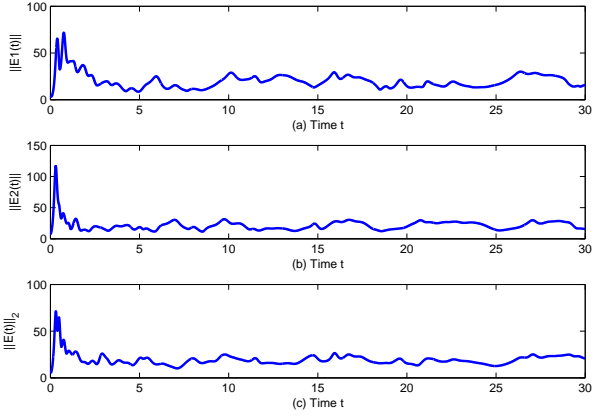}
\caption{The dynamics of $E1(t)$, $E2(t)$ and $\|E(t)\|_2$ under $u_i(t)=0$}\label{adeg5}
\end{center}
\end{figure}

\begin{figure}
\begin{center}
\includegraphics[width=0.5\textwidth]{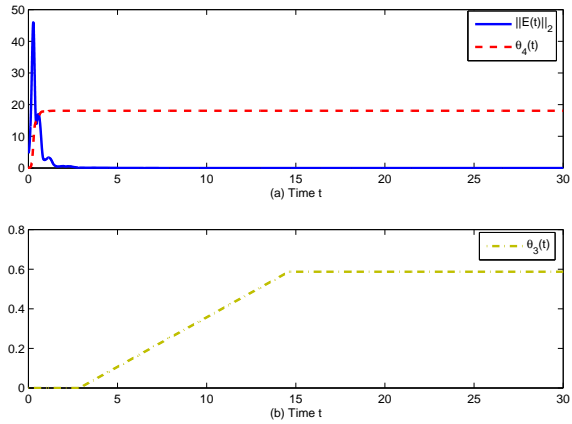}
\caption{The dynamics of $\|E(t)\|_2$, $\theta_3(t)$ and $\theta_4(t)$ under the controller (\ref{comp0}) with adaptive rules (\ref{comp1}) and (\ref{comp2})}\label{adeg6}
\end{center}
\end{figure}

\section{Conclusion}\label{con}
In this paper, we concentrate our attention on FnTSta of delayed systems. Using our proposed 2PM, we design an external control with two terms, FnTSta is realized by using $2$-norm, $1$-norm, and $\infty$-norm, respectively. The requirement of time delay is loose, which can be asynchronous and unbounded. Corresponding convergence rate is also defined using the function $\mu(t)$. AFnTSta is also proved using 2PM under our proposed adaptive rules. Moreover, we apply the obtained theoretical results on network outer synchronization, FnTOSyn and AFnTOSyn are also proved.

We hope this paper may shed some light on FnTSta for delayed systems. Further works include more applications of 2PM, such as FnTSyn under intermittent control.

\section*{Appendix A: Proof of Theorem \ref{1normthm}}
\begin{proof}
2PM will be shown in the next proof procedure. If $\sup_{t-\pi(t)\le s\le t}\|p(s)\|_{1}\le 1, t\ge T$, then we will switch to Phase II directly; else, we will start from Phase I.

{\bf Phase I:} Define a Lyapunov function
\begin{align}\label{v5}
V_5(t)=\mu(t)\|p(t)\|_1=\mu(t)\sum_{i=1}^m|p_i(t)|
\end{align}
and $W_5(t)=\sup_{t-\pi(t)\le s \le t}V_5(s)$. Obviously, $V_5(t)\leq W_5(t)$.

If $V_5(t)<W_5(t)$. Then there must exist a constant $\delta_5>0$ such that $V_5(s)\le W_5(t)$ for $s\in (t,t+\delta_5)$, i.e., $W_5(s)=W_5(t)$ for $s\in (t,t+\delta_5)$.
	
Else if there exists a time $t_5\ge T$, $V_5(t_5)= W_5(t_5)$, then
\begin{align}
&\dot{V}_5(t){|}_{t=t_5}=\dot{\mu}(t)\sum_{i=1}^m|p_i(t)|+\mu(t)\sum_{i=1}^m\mathrm{sgn}(p_i(t))\dot{p}_i(t)\nonumber\\
=&\dot{\mu}(t)\sum_{i=1}^m|p_i(t)|+\mu(t)\sum_{i=1}^m\mathrm{sgn}(p_i(t))[c_1p_i(t)\nonumber\\
&+c_2p_i(t-\pi_i(t))-c_3\mathrm{sgn}(p_i(t))-c_4p_i(t)]\nonumber\\
\le&\bigg[\frac{\dot{\mu}(t)}{\mu(t)}+(c_1-c_4)\bigg]V_5(t)+|c_2|\mu(t)\sum_{i=1}^m|p_{i}(t-\pi_i(t))|\nonumber\\
=&\bigg[\frac{\dot{\mu}(t)}{\mu(t)}+(c_1-c_4)\bigg]V_5(t)\nonumber\\
+&|c_2|\frac{\mu(t)}{\mu(t-\pi(t))}\sum_{i=1}^m\frac{\mu(t-\pi(t))}{\mu(t-\pi_i(t))}\mu(t-\pi_i(t))|p_{i}(t-\pi_i(t))|\nonumber\\
\le&\bigg[\frac{\dot{\mu}(t)}{\mu(t)}+(c_1-c_4)+|c_2|m\frac{\mu(t)}{\mu(t-\pi(t))}\bigg]V_5(t)<0\nonumber
\end{align}

Therefore,
\begin{align}
W_5(T)&\ge W_5(t)=\sup_{t-\pi(t)\le s\le t}\mu(s)\|p(s)\|_1\nonumber\\
&\ge \mu(t-\pi(t))\sup_{t-\pi(t)\le s\le t}\|p(s)\|_1\nonumber
\end{align}
i.e., there exists a point $T_5(\ge T)$,
\begin{align*}
\sup_{t-\pi(t)\le s\le t}\|p(s)\|_1\le \frac{W_5(T)}{\mu(t-\pi(t))}\le 1,~~t\ge T_5
\end{align*}

{\bf Phase II:} From condition (\ref{c2}), we can find a sufficiently small constant $\varepsilon_2>0$, such that $m(|c_2|-c_3)+\varepsilon_2<0$. Using this $\varepsilon_2$, we define another Lyapunov functional
\begin{align}\label{v6}
V_6(t)=\|p(t)\|_1+\varepsilon_2t, t\ge T_1
\end{align}
and $W_6(t)=\sup_{t-\pi(t)\le s\le t}V_6(s), t\ge T_1$.

Obviously, $V_6(t)\le W_6(t)$. If $V_6(t)<W_6(t)$, then there must exist a constant $\delta_6>0$ such that $V_6(s)\le W_6(t)$ for $s\in (t,t+\delta_6)$, i.e., $W_6(s)=W_6(t)$ for $s\in (t,t+\delta_6)$.
	
Else if there exists a point $t_6\ge T_5$, $V_6(t_6)= W_6(t_6)$. Differentiating $V_6(t)$ along (\ref{const}), we have
\begin{align}
&\dot{V}_6(t){|}_{t=t_6}=\sum_{i=1}^m\mathrm{sgn}(p_i(t))[c_1p_i(t)+c_2p_i(t-\pi_i(t))\nonumber\\
&-c_3\mathrm{sgn}(p_i(t))-c_4p_i(t)]+\varepsilon_2\nonumber\\
\le&(c_1-c_4)\|p(t)\|_1+|c_2|\sum_{i=1}^m|p_i(t-\pi_i(t))|-c_3m+\varepsilon_2\nonumber\\
\le&m(|c_2|-c_3)+\varepsilon_2<0\nonumber
\end{align}

Therefore,
\begin{align*}
&\|p(t)\|_1+\varepsilon_2t=V_6(t)\le W_6(t)\le W_6(T_5) \\
=&\sup_{T_5-\pi(T_5)\le s\le T_5}\bigg(\|p(t)\|_1+\varepsilon_2 s\bigg)\le1+\varepsilon_2 T_5
\end{align*}
so $\|p(t)\|_1\le 1-\varepsilon_2(t-T_5)$, i.e., $p(t)\equiv 0$ for any
$t\ge T_6=T_5+{1}/{\varepsilon_2}$.
\end{proof}

\section*{Appendix B: Proof of Theorem \ref{athmnorm1}}
Because of page limit, here we just present corresponding functions, details can be added according to the procedure in Theorem \ref{athm} and calculations in Theorem \ref{1normthm}.

{\bf Phase I}: Choose
\begin{align*}
\overline{V}_5(t)=\mu(t)\|p(t)\|_1+\frac{1}{2d_2}(c_4(t)-c_4^{\star})^2
\end{align*}
whose derivative satisfies
\begin{align*}
\dot{\overline{V}}_5(t)\le\bigg[\frac{\dot{\mu}(t)}{\mu(t)}+(c_1-c_4^{\star})+|c_2|m\frac{\mu(t)}{\mu(t-\pi(t))}\bigg]V_5(t)<0\nonumber
\end{align*}

{\bf Phase II}: Choose
\begin{align*}
&\overline{V}_6(t)\\
=&\|p(t)\|_1+\frac{1}{2d_1}(c_3(t)-c_3^{\star})^2+\frac{1}{2d_3}(c_4(t)-c_4^{\star})^2+\varepsilon_2^{\star}t
\end{align*}
whose derivative satisfies
\begin{align*}
\dot{\overline{V}}_6(t)\le&(c_1-c_4^{\star})\|p(t)\|_1+m(|c_2|-c_3^{\star})+\varepsilon_2^{\star}<0.
\end{align*}

\section*{Appendix C: Proof of Theorem \ref{wqnormthm}}
\begin{proof}
2PM will be shown in the next proof procedure. If $\sup_{t-\pi(t)\le s\le t}\|p(s)\|_{1}\le 1, t\ge T$, then we will switch to Phase II directly; else, we will start from Phase I.

{\bf Phase I:} Define a Lyapunov function
\begin{align}\label{v7}
V_7(t)=\mu(t)\|p(t)\|_{\infty}=\mu(t)\max_{i=1,\cdots,m}|p_i(t)|
\end{align}
and $W_7(t)=\sup_{t-\pi(t)\le s \le t}V_7(s)$. $V_7(t)\leq W_7(t), t\ge T$.
	
If $V_7(t)<W_7(t)$. Then there must exist a constant $\delta_7>0$ such that $V_7(s)\le W_7(t)$ for $s\in (t,t+\delta_7)$, i.e., $W_7(s)=W_7(t)$ for $s\in (t,t+\delta_7)$.
	
Else if there exists a point $t_7\ge T$, $V_7(t_7)= W_7(t_7)$, for this time $t_7$, we suppose that $|p_{i^{\circ}}(t_7)|=\max_{i=1,\cdots,m}|p_i(t_7)|$. Then
\begin{align}
&\dot{V}_7(t){|}_{t=t_7}=\dot{\mu}(t)|p_{i^{\circ}}(t)|+\mu(t)\mathrm{sgn}(p_{i^{\circ}}(t))\dot{p}_{i^{\circ}}(t)\nonumber\\
=&\dot{\mu}(t)|p_{i^{\circ}}(t)|+\mu(t)\mathrm{sgn}(p_{i^{\circ}}(t))[c_1p_{i^{\circ}}(t)\nonumber\\
&+c_2p_{i^{\circ}}(t-\pi_{i^{\circ}}(t))-c_3\mathrm{sgn}(p_{i^{\circ}}(t))-c_4p_{i^{\circ}}(t)]\nonumber\\
\le&\big[\frac{\dot{\mu}(t)}{\mu(t)}+(c_1-c_4)+|c_2|\frac{\mu(t)}{\mu(t-\tau(t))}\big]V_7(t)<0\nonumber
\end{align}

Therefore,
\begin{align}
W_7(T)&\ge W_7(t)=\sup_{t-\pi(t)\le s\le t}\mu(s)\|p(s)\|_{\infty}\nonumber\\
&\ge \mu(t-\pi(t))\sup_{t-\pi(t)\le s\le t}\|p(s)\|_{\infty}\nonumber
\end{align}
i.e., there exists a point $T_7(\ge T)$,
\begin{align*}
\sup_{t-\pi(t)\le s\le t}\|p(s)\|_{\infty}\le \frac{W_7(T)}{\mu(t-\pi(t))}\le 1,~~t\ge T_7
\end{align*}

{\bf Phase II:} From condition (\ref{c2}), we can find a sufficiently small constant $\varepsilon_2>0$, such that $(|c_2|-c_3)+\varepsilon_2<0$. Using this $\varepsilon_2$, we define another Lyapunov function
\begin{align}\label{v8}
V_8(t)=\|p(t)\|_{\infty}+\varepsilon_2t,~~ t\ge T_7
\end{align}
and $W_8(t)=\sup_{t-\pi(t)\le s\le t}V_8(s), t\ge T_7$.

Similarly, $V_8(t)\le W_8(t)$. If $V_8(t)<W_8(t)$, then there must exist a constant $\delta_8>0$ such that $V_8(s)\le W_8(t)$ for $s\in (t,t+\delta_8)$, i.e., $W_8(s)=W_8(t)$ for $s\in (t,t+\delta_8)$.
	
Else if there exists a point $t_8\ge T_7$, $V_8(t_8)= W_8(t_8)$, for this time $t_8$, we suppose that $p_{i^{\diamond}}(t_8)=\max_{i=1,\cdots,m}p_i(t_8)$. Differentiating $V_8(t)$ along (\ref{const}), we have
\begin{align}
&\dot{V}_8(t){|}_{t=t_8}=\mathrm{sgn}(p_{i^{\diamond}}(t))[c_1p_{i^{\diamond}}(t)+c_2p_{i^{\diamond}}(t-\pi_{i^{\diamond}}(t))\nonumber\\
&-c_3\mathrm{sgn}(p_{i^{\diamond}}(t))-c_4p_{i^{\diamond}}(t)]+\varepsilon_2\le(|c_2|-c_3)+\varepsilon_2<0\nonumber
\end{align}

Therefore,
\begin{align*}
&\|p(t)\|_{\infty}+\varepsilon_2t=V_8(t)\le W_8(t)\le W_8(T_7) \\
=&\sup_{T_7-\pi(T_7)\le s\le T_7}\bigg(\|p(t)\|_{\infty}+\varepsilon_2 s\bigg)\le1+\varepsilon_2 T_7
\end{align*}
so $\|p(t)\|_{\infty}\le 1-\varepsilon_2(t-T_7)$, i.e., $p(t)\equiv 0$ for any
$t\ge T_8=T_7+{1}/{\varepsilon_2}$. The proof is completed.
\end{proof}

\section*{Appendix D: Proof of Theorem \ref{athmnormwq}}
Because of page limit, here we just present corresponding functions, details can be added according to the procedure in Theorem \ref{athm} and calculations in Theorem \ref{wqnormthm}.

{\bf Phase I}: Choose
\begin{align*}
\overline{V}_7(t)=\mu(t)\|p(t)\|_{\infty}+\frac{1}{2d_2}(c_4(t)-c_4^{\star})^2
\end{align*}
whose derivative satisfies
\begin{align*}
\dot{\overline{V}}_7(t)\le\bigg[\frac{\dot{\mu}(t)}{\mu(t)}+(c_1-c_4^{\star})+|c_2|\frac{\mu(t)}{\mu(t-\pi(t))}\bigg]V_7(t)<0\nonumber
\end{align*}

{\bf Phase II}: Choose
\begin{align*}
&\overline{V}_8(t)\\
=&\|p(t)\|_{\infty}+\frac{1}{2d_1}(c_3(t)-c_3^{\star})^2+\frac{1}{2d_3}(c_4(t)-c_4^{\star})^2+\varepsilon_2^{\star}t
\end{align*}
whose derivative satisfies
\begin{align*}
\dot{\overline{V}}_8(t)\le&(c_1-c_4^{\star})\|p(t)\|_{\infty}+(|c_2|-c_3^{\star})+\varepsilon_2^{\star}<0.
\end{align*}

\section*{Acknowledgment}
The author would give thanks to Han Zhang, Weijin Liu, Zihan Li and Hailian Ma for their beneficial discussions.

\begin{IEEEbiography}[{\includegraphics[width=1in,height=1.25in,clip,keepaspectratio]{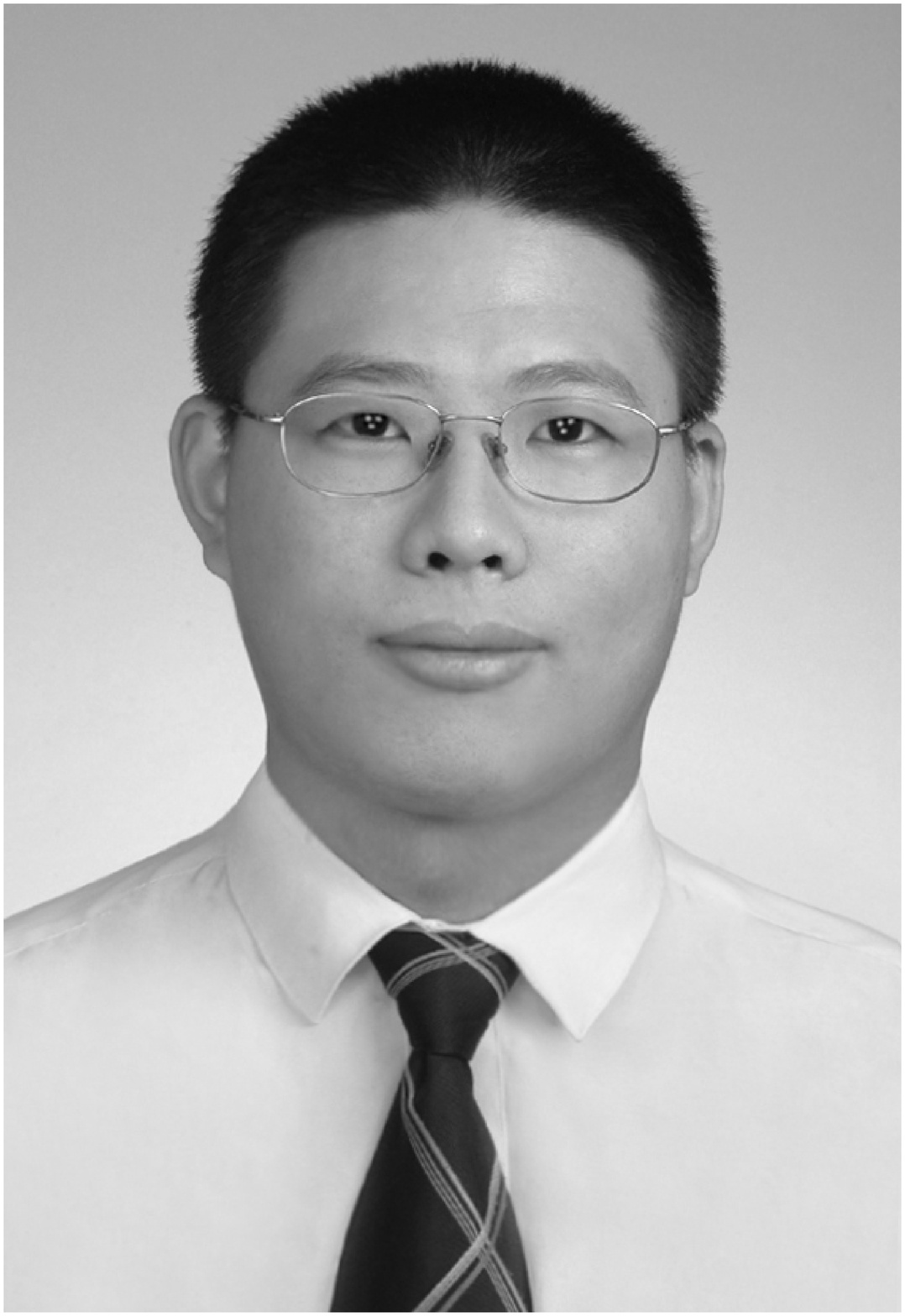}}]{Xiwei Liu}
(M'14, SM'17) received the Ph.D. degree in applied mathematics from Fudan University, Shanghai, China, in 2008. He was a Post-Doctoral Researcher with the Department of Physics, Fudan University, from 2008 to 2010. Then, he joined the Department of Computer Science and Technology, Tongji University, Shanghai. He has been a Professor with the Department of Computer Science and Technology, Tongji University since Dec. 2016. His current research interests include nonlinear dynamical systems, complex networks, and neural networks. He is a recipient of Natural Science Award from Ministry of Education of China in 2015.

He is currently an Associate Editor of \emph{Neurocomputing}.
\end{IEEEbiography}


\begin{thebibliography}{99}
\bibitem{BB1998}
S. P. Bhat and D. S. Bernstein, ``Continuous finite-time stabilization of the translational and rotational double integrators,'' \emph{IEEE Trans. Autom. Control}, vol. 43, no. 5, pp. 678-682, May 1998.

\bibitem{BB2000}
S. P. Bhat and D. S. Bernstein, ``Finite-time stability of continuous autonomous systems,'' {\emph{SIAM J. Control Optim.}}, vol. 38, no. 3, pp. 751-766, Mar. 2000.

\bibitem{hongyiguang2002}
Y. G. Hong, Y. S. Xu, and J. Huang, ``Finite-time control for robot manipulators,'' \emph{Systems $\&$ Control Letters}, vol. 46, no. 4, pp. 243-253, Jul. 2002.

\bibitem{MP2008}
E. Moulay and W. Perruquetti, ``Finite time stability conditions for non-autonomous continuous systems,'' {\emph{Int. J. Control}}, vol. 81, no. 5, pp. 797-803, May 2008.

\bibitem{linwei2017}
Y. Z. Sun, S. Y. Leng, Y. C. Lai, C. Grebogi, and W. Lin, ``Closed-loop control of complex networks: a trade-off between time and energy,'' \emph{Phys. Rev. Lett.}, vol. 119, no. 19, p. 198301, Nov. 2017.

\bibitem{P2012}
A. Polyakov, ``Nonlinear feedback design for fixed-time stabilization of linear control systems,'' \emph{IEEE Trans. Autom. Control}, vol. 57, no. 8, pp. 2106-2110, Aug. 2012.

\bibitem{Murray2004}
R. Olfati-Saber and R. M. Murray, ``Consensus problems in networks of agents with switching topology and time-delays,'' \emph{IEEE Trans. Autom. Control}, vol. 49, no. 9, pp. 1520-1533, Sep. 2004.

\bibitem{Chenliulu2007}
T. Chen, X. Liu, and W. Lu, ``Pinning complex networks by a single controller,'' \emph{IEEE Transactions on Circuits and Systems I-Regular Papers}, vol. 54, no. 6, pp. 1317-1326, Jun. 2007.

\bibitem{c2006}
J. Cort\'{e}s, ``Finite-time convergent gradient flows with applications to networks consensus,'' \emph{Automatica}, vol. 42, no. 11, pp. 1993-2000, Nov. 2006.

\bibitem{jiangwang2011}
F. C. Jiang and L. Wang, ``Finite-time weighted average consensus with respect to a monotonic function and its application,'' \emph{Syst. Control Lett.}, vol. 60, no. 9, pp. 718-725, Sep. 2011.


\bibitem{yangtac2016}
X. S. Yang and J. Q. Lu, ``Finite-time synchronization of coupled networks with markovian topology and impulsive effects,'' \emph{IEEE Trans. Autom. Control}, vol. 61, no. 8, pp. 2256-2261, Aug. 2016.


\bibitem{licaojiang2019}
H. L. Li, J. D. Cao, H. J. Jiang, and A. Alsaedi, ``Finite-time synchronization and parameter identification of uncertain fractional-order complex networks,'' \emph{Physica A}, vol. 533, p. 122027, Nov. 2019.

\bibitem{p2015}
A. Polyakov, D. Efimov, and W. Perruquetti, ``Finite-time and fixed-time observer design: Implicit Lyapunov function approach,'' \emph{Automatica}, vol. 51, pp. 332-340, Jan. 2015.

\bibitem{luliuchen2016}
W. L. Lu, X. W. Liu, and T. P. Chen, ``A note on finite-time and fixed-time stability,'' \emph{Neural Networks}, vol. 81, pp. 11-15, Sep. 2016.

\bibitem{liuchen2018}
X. W. Liu and T. P. Chen, ``Finite-time and fixed-time cluster synchronization with or without pinning control,'' \emph{IEEE T. Cybern.}, vol. 48, no. 1, pp. 240-252, Jan. 2018.

\bibitem{huyunn2017}
C. Hu, J. Yu, Z. H. Chen, H. J. Jiang, and T. W. Huang, ``Fixed-time stability of dynamical systems and fixed-time synchronization of coupled discontinuous neural networks,'' \emph{Neural Netw.}, vol. 89, pp. 74-83, May 2017.

\bibitem{liuxiaoyang2019}
X. Y. Liu, D. W. C. Ho, Q. Song, and W. Y. Xu, ``Finite/fixed-time pinning synchronization of complex networks with stochastic disturbances,'' \emph{IEEE T. Cybern.}, vol. 49, no. 6, pp.2398-2403, Jun. 2019.

\bibitem{zuoii2018}
Z. Y. Zuo, Q. L. Han, B. D. Ning, X. H. Ge, and X. M. Zhang, ``An overview of recent advances in fixed-time cooperative control of multiagent systems,'' \emph{IEEE Trans. Ind. Inform.}, vol. 14, no. 6, pp. 2322-2334, Jun. 2018.

\bibitem{rios2017}
H. Rios, D. Efimov, J. A. Moreno, W. Perruquetti, and J. G. Rueda-Escobedo, ``Time-varying parameter identification algorithms: finite and fixed-time convergence,'' {\emph{IEEE Trans. Autom. Control}}, vol. 62, no. 7, pp. 3671-3678, Jul. 2017.

\bibitem{liuchen2016}
X. W. Liu and T. P. Chen, ``Global exponential stability for complex-valued recurrent neural networks with asynchronous time delays,'' \emph{IEEE Trans. Neural Netw. Learn. Syst.}, vol. 27, no. 3, pp. 593-606, Mar. 2016.

\bibitem{chenmu}
T. P. Chen and L. L. Wang, ``Global $\mu$-stability of delayed neural networks with unbounded time-varying delays,'' \emph{IEEE Trans. Neural Netw.}, vol. 18, no. 6, pp. 1836-1840, Nov. 2007.

\bibitem{wancaonn2016}
Y. Wan, J. D. Cao, G. H. Wen, and W. W. Yu, ``Robust fixed-time synchronization of delayed Cohen-Grossberg neural networks,'' {\emph{Neural Netw.}}, vol. 73, pp. 86-94, Jan. 2016.

\bibitem{liujianghu2017}
M. Liu, H. J. Jiang, and C. Hu, ``Finite-time synchronization of delayed dynamical networks via aperiodically intermittent control,'' {\emph{J. Frankl. Inst.-Eng. Appl. Math.}}, vol. 354, no. 13, pp. 5374-5397, Sep. 2017.

\bibitem{guogonghuang2018}
Z. Y. Guo, S. Q. Gong, and T. W. Huang, ``Finite-time synchronization of inertial memristive neural networks with time delay via delay-dependent control,'' {\emph{Neurocomputing}}, vol. 293, pp. 100-107, Jun. 2018.

\bibitem{weicaocn2018}
R. Y. Wei, J. D. Cao, A. Alsaedi, ``Finite-time and fixed-time synchronization analysis of inertial memristive neural networks with time-varying delays,'' {\emph{Cogn. Neurodynamics}}, vol. 12, no. 1, pp. 121-134, Feb. 2018.


\bibitem{ganxaosheng2019}
Q. T. Gan, F. Xiao, and H. Sheng, ``Fixed-time outer synchronization of hybrid-coupled delayed complex networks via periodically semi-intermittent control,'' {\emph{J. Frankl. Inst.-Eng. Appl. Math.}}, vol.356, no. 12, pp. 6656-6677, Aug. 2019.


\bibitem{yanli2019}
L. H. Yan and J. M. Li, ``Adaptive aperiodically intermittent synchronization for complex dynamical network with unknown time-varying outer coupling strengths,'' {\emph{Math. Probl. Eng.}}, p. 6732357, 2019.

\bibitem{jingzhangmeifan2019}
T. Y. Jing, D. Y. Zhang, J. Mei, and Y. L. Fan, ``Finite-time synchronization of delayed complex dynamic networks via aperiodically intermittent control,'' {\emph{J. Frankl. Inst.-Eng. Appl. Math.}}, vol.356, no. 10, pp. 5464-5484, Jul. 2019.

\bibitem{liuqin2019}
Y. J. Liu, Y. Qin, J. J. Huang, T. W. Huang, and X. B. Yang, ``Finite-time synchronization of complex-valued neural networks with multiple time-varying delays and infinite distributed delays,'' {\emph{Neural Process. Lett.}}, vol. 50, no. 2, pp. 1773-1787, Oct. 2019.



\bibitem{yangprop}
X. L. Xiong, R. Q. Tang, and X. S. Yang, ``Finite-time synchronization of memristive neural networks with proportional delay,'' {\emph{Neural Processing Letters}}, vol. 50, pp. 1139-1152, Oct. 2019.

\bibitem{wangchen2018}
L. L. Wang and T. P. Chen, ``Finite-time anti-synchronization of neural networks with time-varying delays,'' \emph{Neurocomputing}, vol. 275, pp. 1595-1600, Jan. 2018.

\bibitem{wangchen2019}
L. L. Wang and T. P. Chen, ``Finite-time and fixed-time anti-synchronization of neural networks with time-varying delays,'' \emph{Neurocomputing}, vol. 329, pp. 165-171, Feb. 2019.

\bibitem{liuli2020}
X. W. Liu and Z. H. Li, ``Finite time anti-synchronization of complex-valued neural networks with bounded asynchronous time-varying delays,'' {\emph{Neurocomputing}}, doi:10.1016/j.neucom.2020.01.035.

\end{thebibliography}
\end{document}